\documentclass[twoside,11pt,theapa]{article}
\usepackage{jair, rawfonts}
\usepackage[natbibapa]{apacite}
\let\cite\relax
\let\cite\citep
\usepackage{algorithm}
\usepackage{algorithmic}
\usepackage{ifthen}

\usepackage[inline,shortlabels]{enumitem}
\usepackage{newfloat}
\usepackage{listings}
\floatstyle{ruled}
\newfloat{listing}{tb}{lst}{}
\floatname{listing}{Listing}
\usepackage[usenames,dvipsnames]{xcolor}

\usepackage[pdftex, colorlinks=true, hyperfootnotes=true, hyperindex=true,
            plainpages=false, pagebackref=false, pdfpagelabels=true, pdfstartview=FitH,
            linkcolor=black, citecolor=black, urlcolor=black,
            bookmarks, bookmarksopen, bookmarksdepth=3]{hyperref}

\usepackage[utf8]{inputenc}
\usepackage{todonotes}
\presetkeys{todonotes}{inline,backgroundcolor=yellow}{}
\usepackage{xspace}
\usepackage{amsmath,amssymb,amsfonts}
\usepackage{amsthm}
\allowdisplaybreaks
\usepackage{stmaryrd}
\usepackage{comment}
\usepackage[scientific-notation=false,group-separator={,}]{siunitx}%
\usepackage{adjustbox}
\usepackage{booktabs}
\usepackage{array}
\usepackage{color, colortbl}

\usepackage{xcolor}

\newcommand\cdf[1]{{\color{black}{#1}}}

\newcommand\cg[1]{{\color{black}{#1}}}

\newenvironment{marcodo}[1][]{%
\color{black}}{\normalcolor}
\newenvironment{claudiodo}[1][]{%
\color{black}}{\normalcolor}
\newcommand{\mr}[1]{\begin{marcodo}#1\end{marcodo}}
\newcommand{\cdc}[1]{\begin{claudiodo}#1\end{claudiodo}}

\newcommand{\declare}{\textsc{Declare}\xspace}
\newcommand{\LTLf}{LTL\textsubscript{f}\xspace}
\newcommand{\LTL}{LTL\xspace}
\newcommand{\nusmv}{\textsc{NuSMV}\xspace}
\newcommand{\trppp}{\textsc{trp++}\xspace}
\newcommand{\aaltaf}{\textsc{aaltaf}\xspace}
\newcommand{\UC}{\ensuremath{\mathrm{UC}}\xspace}
\newcommand{\UCs}{\ensuremath{\mathrm{UC}\textrm{s}}\xspace}
\newcommand{\UCset}{\ensuremath{\mathrm{UCS}}}

\DeclareMathOperator{\LPre}{\mathbf{Y}}
\DeclareMathOperator{\LWpre}{\mathbf{Z}}
\DeclareMathOperator{\LOnce}{\mathbf{O}}
\DeclareMathOperator{\LHistorically}{\mathbf{H}}
\DeclareMathOperator{\LSince}{\mathbf{S}}
\DeclareMathOperator{\LTrigger}{\mathbf{T}}

\DeclareMathOperator{\LNext}{\mathbf{X}}
\DeclareMathOperator{\LWnext}{\begin{marcodo}\mathbf{N}\end{marcodo}}
\DeclareMathOperator{\LGlobally}{\mathbf{G}}
\DeclareMathOperator{\LEventually}{\mathbf{F}}
\DeclareMathOperator{\LRelease}{\mathbf{R}}
\DeclareMathOperator{\LUntil}{\mathbf{U}}
\newcommand{\opt}{\ensuremath{\;\;\vert\;\;}}
\newcommand{\PVarSet}{\mathcal{V}}
\newcommand{\len}[1]{\ensuremath{\mathsf{len}(#1)}}
\newcommand{\ftol}[1]{\ensuremath{\mathsf{f2l}(#1)}}
\newcommand{\ltlftoltl}[1]{\ensuremath{\mathsf{LTLf2LTL}(#1)}}
\newcommand{\ltlptoltl}[1]{\ensuremath{\mathsf{P2F}(#1)}}
\newcommand{\xnf}[1]{\ensuremath{\mathsf{xnf}(#1)}}
\newcommand{\lend}{\ensuremath{\mathsf{end}}\xspace}
\newcommand{\den}[1]{\ensuremath{\llbracket #1 \rrbracket}}
\newcommand{\pair}[1]{\ensuremath{\langle #1 \rangle}}

\newtheorem{theorem}{Theorem}
\newtheorem{corollary}{Corollary}
\newtheorem{lemma}{Lemma}
\newtheorem{definition}{Definition}

\makeatletter
\newcommand{\IWHILE}[1]{\ALC@it\algorithmicwhile\ #1\ \algorithmicdo\begin{ALC@whl}}
\newcommand{\IENDWHILE}{\end{ALC@whl}}
\newcommand{\IIF}[1]{\ALC@it\algorithmicif\ #1\ \algorithmicthen}
\newcommand{\ENDIIF}{}
\newcommand{\IRETURN}{\algorithmicreturn{} \ }
\newcommand\Call[2]{\textsc{#1}\ifthenelse{\equal{#2}{}}{}{(#2)}}%
\makeatother

\title{
	Computing unsatisfiable cores for \LTLf specifications
}

\author{%
	\name{Marco Roveri} \email{marco.roveri@unitn.it} \\
	\addr{University of Trento, \\ Via Sommarive 9, Trento, Italy}
	\AND
	\name{Claudio Di Ciccio} \email{claudio.diciccio@uniroma1.it} \\ 
	\addr{Sapienza University of Rome, \\ Viale Regina Elena 295, 00161 Rome, Italy}
	\AND
	\name{Chiara Di Francescomarino} \email{dfmchiara@fbk.eu} \\
	\name{Chiara Ghidini} \email{ghidini@fbk.eu} \\
	\addr{FBK IRST, \\ Via Sommarive 18, Trento, Italy} %
}

\begin{document}

\maketitle

\begin{abstract}
	Linear-time temporal logic on finite traces (\LTLf) is rapidly becoming a de-facto standard to produce specifications in many application domains (e.g., planning, business process management, run-time monitoring, reactive synthesis). Several studies approached the respective satisfiability problem.
In this paper, we investigate the problem of extracting the unsatisfiable core in \LTLf specifications.
We provide four algorithms for extracting an unsatisfiable core leveraging the adaptation of state-of-the-art approaches to \LTLf satisfiability checking.
We implement the different approaches within the respective tools and carry out an experimental evaluation on a set of reference benchmarks, restricting to the  unsatisfiable ones.
The results show the feasibility, effectiveness, and complementarities of the different algorithms and tools.
%

\end{abstract}



\section{Introduction}



A growing body of literature evidences the adoption of
linear-time temporal logic on finite traces (\LTLf)~\cite{DBLP:conf/ijcai/GiacomoV13} 
to produce systems specifications~\cite{DBLP:conf/aaai/GiacomoMM14}.
Its widespread use spans across several application domains, including
business process management (BPM) for declarative process modeling~\cite{DBLP:journals/tweb/MontaliPACMS10,DBLP:conf/bpm/GiacomoMGMM14}
\begin{claudiodo}
and mining~\cite{DBLP:conf/bpm/CecconiCGM18,DBLP:conf/otm/RaimCMMM14}%
\end{claudiodo},
run-time monitoring and verification~\cite{DBLP:conf/bpm/GiacomoMGMM14,DBLP:journals/corr/abs-2004-01859,8133351},
%
and AI planning~\cite{DBLP:conf/kr/CalvaneseGV02,DBLP:conf/aaai/SohrabiBM11,DBLP:conf/aips/CamachoBMM18,DBLP:conf/ijcai/CamachoM19}.

\cg{When it comes to verification techniques and tool support for \LTLf,} several studies approach the \LTLf satisfiability problem via
reduction to \LTL~\cite{DBLP:conf/focs/Pnueli77} satisfiability on infinite traces~\cite{DBLP:conf/aaai/GiacomoMM14}, or via specific
propositional satisfiability approaches~\cite{DBLP:journals/jair/FiondaG18,DBLP:journals/ai/LiPZVR20}.
However, 
no efforts have been devoted thus far to the identification of the formulas that lead to unsatisfiability in \LTLf specifications, with the consequence that no support has been offered
for modelers and system designers to single out the causes of possible inconsistencies.

In this paper, we tackle the challenge of extracting unsatisfiable cores (\UCs) from \LTLf specifications. \cg{%
Investigating this problem is interesting both from practical and theoretical viewpoints.
On the practical side, if unsatisfiability signals that a specification is defective, the identification of unsatisfiable cores provides the users with the opportunity to isolate the source of inconsistency and leads them to a consequent debugging.
Notice that determining a reason for unsatisfiability without automated support may reveal unfeasible for a number of reasons that range from the sheer size of the formula to the lack of time and skills of the user~\cite{DBLP:journals/scp/Schuppan12,DBLP:conf/ictai/Schuppan18}.
On the theoretical side, we remark that dealing with the extraction of \UCs in \LTLf specifications is far from trivial.
Indeed, there is neither a default pathway to move from the support provided for \LTL to the one that has to be provided for \LTLf nor a default algorithm upon which this transition could be based.
Concerning the \emph{pathway}, there are two clear alternatives to address this problem: the first one extends techniques for the extraction of \UCs in \LTL to the case of \LTLf; the second one exploits algorithms that directly compute satisfiability in \LTLf to provide support for the extraction of \UCs.
Concerning the specific \emph{algorithms} from which to start, the two approaches present different scenarios. In the first pathway, it is easy to observe that several techniques for the extraction of \UCs in \LTL exist, and could be extended to the case of \LTLf. Since recent works show that a single universal best algorithm does not exist and often the systems exhibit behaviors that complement each other~\cite{DBLP:journals/fmsd/LiZPZV19,DBLP:journals/ai/LiPZVR20}, choosing a single algorithm from which to start is less than obvious. In the second pathway, instead, 
the number of works on satisfiability in \LTLf is still rather limited.

In this work, we 
explore both the above pathways. 
For the \LTL pathway, in particular, we consider algorithms belonging to two reference approaches: one based on model-checking, and the other one based on theorem proving. 
\mr{For the \LTLf pathway, we consider a reference state-of-the-art specific reduction to propositional satisfiability.}
We believe that leveraging reference state-of-the-art approaches provides a rich starting point for the investigation of the problem and the provision of effective tools for the extraction of \UCs in \LTLf specifications.
Our comparative evaluation shows a complementary behavior of the different algorithms.
}
%

\mr{Our contributions consist of the following}:
\cg{\begin{enumerate}
\item Four algorithms that allow for the computation of an unsatisfiable core through the adaptation of the main reference state-of-the-art approaches for \LTL and \LTLf satisfiability checking (Section~\ref{sec:algorithms}). For the \LTL pathway, we consider two satisfiability checking algorithms: one based on Binary Decision Diagrams (BDDs)~\cite{DBLP:journals/fmsd/ClarkeGH97}, and the other based on propositional satisfiability~\cite{DBLP:journals/lmcs/BiereHJLS06}, and a theorem proving algorithm based on temporal resolution~\cite{DBLP:conf/cade/HustadtK03,DBLP:journals/acta/Schuppan16}. For the native \LTLf pathway we consider the reference work of~\cite{DBLP:journals/ai/LiPZVR20} \mr{based on explicit search and} propositional satisfiability. Note that, the techniques based on propositional satisfiability (that is, based on~\cite{DBLP:journals/lmcs/BiereHJLS06} and~\cite{DBLP:journals/ai/LiPZVR20})
aim at extracting \mr{an} \UC, which may not necessarily be the minimum one.
The BDD and temporal-resolution based algorithms already allow for the extraction of a minimum unsatisfiable core\cdf{.}
\item An implementation of the proposed four algorithms (Section~\ref{sub:implementation}). Three implementations extend existing tools for the corresponding original algorithms; the implementation of \mr{the}
algorithm \cdf{based on temporal resolution,} instead, \cdc{resorts to} a pre-processing of the formula to reduce the input to the language restrictions of the original tool.
\item An experimental evaluation on a \mr{large} set of reference benchmarks taken from~\cite{DBLP:journals/ai/LiPZVR20}, restricted to the unsatisfiable ones (Sections~\ref{sub:the_experimental_setup} and~\ref{sub:results}).
The results show an overall better \cdc{time efficiency} of the algorithm based on the native \LTLf pathway \cite{DBLP:journals/ai/LiPZVR20}. \cdc{However, the cardinality of the \UC extracted by the fastest approach is the smallest one in only about half of the cases. The experimental findings exhibit} a complementarity of the proposed approaches on different specifications: %
\cdc{depending on the varying number of \mr{propositional} variables, number of conjuncts and degree of nesting of the temporal operators in the benchmarks, it is not rare that \mr{some of} the implemented techniques achieve a noticeable performance when \mr{the} other \mr{ones} terminate with no result and vice-versa.}
\end{enumerate}
}

\cg{Since popular usages of \LTLf \mr{leverage} past temporal operators (see e.g., the \declare language~\cite{2009-Aalst}), we also} provide a way to handle \LTLf with past temporal operators \cg{(see Definition~\ref{def:LTLf-formulasat} and all the \mr{respective} technical parts). This} results in the same expressive power as that of pure future version, though allowing
for exponentially more succinct specifications~\cite{DBLP:conf/tls/Gabbay87,DBLP:conf/lics/LaroussinieMS02} and more natural encodings of \LTLf based modeling languages \cg{that make use of these operators}.
This objective is pursued by leveraging algorithms already supporting \LTL with past
temporal operators, or through a reduction to \LTLf with only future
temporal operators to use existing approaches for \LTLf
satisfiability checking.

\cg{The contributions highlighted above are complemented with sections where we illustrate relevant background knowledge (Section~\ref{sec:background}), related works (Section~\ref{sec:related}), conclusions and future work (Section~\ref{sec:conclusions}).}
%



\section{Background}
\label{sec:background}
\sloppypar


\subsection{\LTLf Syntax and Semantics}

We assume that a finite set of propositional variables $\PVarSet$ is
given.

A \emph{state} $\mu$ over propositional variables in $\PVarSet$ is a
complete assignment of a Boolean value to variables in $\PVarSet$.
\begin{marcodo}
\begin{definition}
  We say that variable $x \in \PVarSet$ holds in $\mu$ iff $x$ is
  assigned the true value in $\mu$, and we denote this as
  $\mu \models_p x$.
\end{definition}
\end{marcodo}

A \emph{finite trace} over propositional variables in $\PVarSet$ is a
sequence $\pi = \mu_0, \mu_1, ..., \mu_{n-1}$ of states. The length of
a trace $\pi = \mu_0, \mu_1, ..., \mu_{n-1}$, denoted $\len{\pi}$,
is $n$. We denote with $\pi[i]$ the $i$-th state $\mu_i$, and with
$\pi[i:-]$ the suffix of the finite trace starting at state $i$, i.e.,
$\pi[i:-] = \mu_i, ..., \mu_{n-1}$.
\begin{marcodo}
  An \emph{infinite trace} over propositional variables in $\PVarSet$
  is a sequence $\pi = \mu_0, \mu_1, ...$ of states such that
  $\pi \in (2^{\PVarSet})^{\omega}$. Given two finite traces $\pi_1$
  and $\pi_2$, we indicate with $\pi_1\pi_2^{\omega}$ the infinite
  \emph{lazo-shaped} trace with prefix $\pi_1$ and trace $\pi_2$
  repeated indefinitely (intuitively, to indicate that $\pi_2$ is
  repeated within an infinite loop).
\end{marcodo}

An \emph{\LTLf formula} $\varphi$ is built over the propositional
variables in $\PVarSet$ by using the classical Boolean connectives
``$\wedge$'', ``$\vee$'', and ``$\neg$'', complemented with the future temporal operators ``$\LNext$'' (next), ``$\LWnext$'' (weak
next), ``$\LGlobally$'' (always/globally), ``$\LEventually$''
(eventually/finally), ``$\LUntil$'' (until) and ``$\LRelease$''
(release), and with the past temporal operators ``$\LPre$'' (yesterday), ``$\LWpre$'' (weak
yesterday), ``$\LHistorically$'' (historically), ``$\LOnce$''
(once), ``$\LSince$'' (since), and ``$\LTrigger$'' (trigger).
The $\LWnext$ (resp. $\LWpre$) operator is similar to $\LNext$ ($\LPre$) and solely
differs in the way the final (resp. initial) state is dealt with:
In the last (resp.\ initial) state, $\LNext \varphi$ ($\LPre \varphi$) is false, while $\LWnext \varphi$ (resp. $\LWpre \varphi$) is true.
The grammar for building \LTLf formulas is:
\begin{marcodo}
\begin{align*}
  \varphi ::= & \;\; x \opt (\varphi_1 \wedge \varphi_2) \opt (\varphi_1 \vee \varphi_2) \opt \neg \varphi_1 \opt \\
              & \text{\textit{Future temporal operators}}\\
              & (\LNext \varphi_1) \opt (\LWnext \varphi_1) \opt (\LEventually \varphi_1) \opt (\LGlobally \varphi_1) \opt
               (\varphi_1 \LUntil \varphi_2) \opt (\varphi_1 \LRelease \varphi_2),\\
              & \text{\textit{Past temporal operators}}\\
              & (\LPre \varphi_1) \opt (\LWpre \varphi_1) \opt (\LHistorically \varphi_1) \opt (\LOnce \varphi_1) \opt
               ( \varphi_1 \LSince \varphi_2) \opt ( \varphi_1 \LTrigger \varphi_2),
\end{align*}
\end{marcodo}%
\noindent where $x\in \PVarSet$ is a propositional variable,
\begin{marcodo} $\varphi_1$ and $\varphi_2$ are \LTLf formulas.\end{marcodo}
Classical implication and equivalence connectives can be obtained in
standard ways in terms of the $\wedge,\vee,\neg$ connectives.

\begin{definition}
	\label{def:LTLf-formulasat}
Given a finite trace $\pi$, the \LTLf formula $\varphi$ is true
in $\pi$ at state $\pi[i]$ s.t. $i \in [0..\len{\pi}-1]$, denoted with
$\pi,i \models \varphi$, iff:
\begin{itemize}[itemsep=1pt]
\item $\pi,i \models x$ iff $\pi[i] \models_p x$;
\item $\pi,i \models \varphi_1 \wedge \varphi_2$ iff $\pi,i \models \varphi_1$ and $\pi,i \models \varphi_2$;
\item $\pi,i \models \varphi_1 \vee \varphi_2$ iff $\pi,i \models \varphi_1$ or $\pi,i \models \varphi_2$;
\end{itemize}
Future temporal operators:
\begin{itemize}
\item $\pi,i \models \LNext \varphi$ iff $i < \len{\pi} -1$ and $\pi,i+1 \models \varphi$;
\item $\pi,i \models \LWnext \varphi$ iff $i < \len{\pi} -1$ and
  $\pi,i+1 \models \varphi$, or $i = \len{\pi}-1$;
\item $\pi,i \models \LEventually \varphi$ iff
  $\exists j \textrm{ with } i \le j < \len{\pi}$ such that
  $\pi,j \models \varphi$;
\item $\pi,i \models \LGlobally \varphi$ iff
  $\forall j \textrm{ with } i \le j < \len{\pi}$ it holds
  $\pi,j \models \varphi$;
\item $\pi,i \models \varphi_1 \LUntil \varphi_2$ iff
  $\exists j \textrm{ with } i \le j < \len{\pi}$ such that
  $\pi,j \models \varphi_2$ and
  $\forall k \textrm{ with } i \le k < j$ it holds that
  $\pi,k \models \varphi_1$;
\item $\pi,i \models \varphi_1 \LRelease \varphi_2$ iff
  $\forall j \textrm{ with } i \le j < \len{\pi}$ it holds that
  $\pi,j \models \varphi_2$, or
  $\exists j \textrm{ with } i \le j < \len{\pi}$ such that
  $\pi,j \models \varphi_1$ and
  $\forall k \textrm{ with } i \le k \begin{marcodo}\le\end{marcodo} j$ it holds that
  $\pi,k \models \varphi_2$;
\end{itemize}
Past temporal operators:
  \begin{itemize}
  \item $\pi,i \models \LPre \varphi$ iff $1 \le i$ and $\pi, i-1 \models \varphi$;
  \item $\pi,i \models \LWpre \varphi$ iff $0 = i$ or $\pi, i-1 \models \varphi$;
  \begin{marcodo}
  \item $\pi,i \models \LOnce \varphi$ iff $\exists j \text{ with } 0 \le j \le i$ such that $\pi, j \models \varphi$;
  \item $\pi,i \models \LHistorically \varphi$ iff $\forall j \text{ with } 0 \le j \le i$ it holds that $\pi, j \models \varphi$;
  \item $\pi,i \models \varphi_1 \LSince \varphi_2$ iff $\exists k \textrm{ with } 0 \le k \le i$ \textrm{ such that } $\pi, k \models \varphi_2 \textrm{ and } \forall j \textrm{ with } k < j \le i$ it holds that
$\pi, j \models \varphi_1$;
  \item $\pi,i \models \varphi_1 \LTrigger \varphi_2$ iff $\forall k \textrm{ with } 0 \le k \le i$ \textrm{ such that } $\pi, k \models \varphi_2 \textrm{ or } \exists j \textrm{ with } k < j \le i$ such that $\pi, j \models \varphi_1$;
  \end{marcodo}
  \end{itemize}
We say that $\pi$ is a \emph{model of $\varphi$} whenever
$\pi,0 \models \varphi$, and we say that $\varphi$ \emph{is
  satisfiable}, whenever there exists a $\pi$ such that
$\pi,0 \models \varphi$.
\end{definition}

\mr{

\paragraph{Remark.} Notice that the following equivalences hold:
$(\varphi_1 \LTrigger \varphi_2) \leftrightarrow \neg (\neg
\varphi_1 \LSince \neg \varphi_2)$,
$(\LWpre \varphi) \leftrightarrow \neg (\LPre \neg \varphi)$, and
$(\varphi_1 \LRelease \varphi_2) \leftrightarrow \neg (\neg
\varphi_1 \LUntil \neg \varphi_2)$. In the following we leverage
these equivalences whenever needed to simplify the presentation and the proofs.

}

The language of an \LTLf formula $\varphi$ over the set of $\PVarSet$
is defined as the set
$\mathcal{L}(\varphi) := \{ \pi | \pi,0 \models \varphi\}$.
Thus, the satisfiability problem for an \LTLf formula $\varphi$ can be
reduced to checking that $\mathcal{L}(\varphi) \not=\emptyset$.

Let us consider the formula $\LGlobally (a \rightarrow \LWnext b)$.
Trace $\pi^1 = \mu^1_0, \mu^1_1, \mu^1_2, \mu^1_3$ of length 4
such that $\mu^1_0 = \{a=\bot,b=\top\}$,
$\mu^1_1 = \{a=\top,b=\bot\}$, $\mu^1_2 = \{a=\top,b=\top\}$,
$\mu^1_3 = \{a=\top,b=\top\}$ satisfies the formula,
while $\pi^2 = \mu^2_0, \mu^2_1, \mu^2_2, \mu^2_3$ of length 4 such
that $\mu^2_0 = \{a=\bot,b=\top\}$, $\mu^2_1 = \{a=\top,b=\bot\}$,
$\mu^2_2 = \{a=\top,b=\top\}$, $\mu^2_3 = \{a=\top,b=\bot\}$ is not a
model since \begin{marcodo}$\mu^2_2 \models_p a$\end{marcodo} but in the next state
$\mu^2_3 \not \models_p b$.
\begin{marcodo}On the other hand,\end{marcodo} the \LTLf formula
$\LGlobally (a \rightarrow \LNext b)$ does not hold for both traces:
\begin{marcodo} $\pi_1$ does not satisfy the formula because\end{marcodo} in the
last state $\mu^1_3 \models_p a$ and there is no next state; %
\begin{marcodo} as for $\pi_2$, $\mu^2_2 \models_p a$ but in
  the next state $\mu^2_3 \not \models_p b$, and 
  $\mu^1_3 \models_p a$ but that is the last state, so no next state exists.\end{marcodo}



\subsubsection{Unsatisfiable core}

Given a set $\Gamma = \{\varphi_1, ..., \varphi_N\}$ of \LTLf formulas
$\varphi_i$ (considered in implicit conjunction, i.e.
$\Gamma = \bigwedge_{i=0}^{N} \varphi_i$), such that $\Gamma$ is not
satisfiable, we say that a formula $\Phi \subseteq \Gamma$ is an
\emph{unsatisfiable core} of $\Gamma$ iff $\Phi$ is unsatisfiable.
A \emph{minimal unsatisfiable core} $\Phi$ is such that each \LTLf
formula $\Phi_i = \Phi \setminus \{\varphi_i\}$ for $\varphi_i \in \Phi$
is satisfiable.
A \emph{minimum unsatisfiable core} is a minimal unsatisfiable core
with the smallest possible cardinality.

\subsection{Checking Satisfiability of an \LTLf Formula}

Checking the satisfiability of an \LTLf formula $\varphi$ can be reduced to
checking language emptiness of a Nondeterministic Finite state
Automaton (NFA)~\cite{DBLP:conf/aaai/GiacomoMM14}.
Alternative approaches for \LTLf formulas without past temporal
operators~\cite{DBLP:conf/ijcai/GiacomoV13,DBLP:conf/aaai/GiacomoMM14,DBLP:journals/jair/FiondaG18}
address this problem by checking the satisfiability of an
equi-satisfiable \LTL formula over infinite traces\footnote{We refer
  the reader to \cite{DBLP:conf/focs/Pnueli77,Tsay2021} for the semantics of
  \LTL over infinite traces.%
} leveraging on existing well established techniques (e.g
\cite{DBLP:journals/fmsd/ClarkeGH97,DBLP:journals/lmcs/BiereHJLS06}).
These approaches proceed as follows:
\begin{enumerate*}[label=(\roman*)]
\item they introduce a new fresh propositional variable
  $\lend \not\in \PVarSet$ used to denote the trace has ended;
\item they require that \lend eventually holds (i.e.,
  $\LEventually \lend$);
\item they require that once \lend becomes true, it stays true forever
(i.e. $\LGlobally (\lend \rightarrow \LNext \lend)$);
\item they translate the \LTLf formula $\varphi$ into an \LTL formula by means of
  a translation function $\ftol{\varphi}$ that is defined recursively on
  the structure of the \LTLf formula $\varphi$ as follows:
\end{enumerate*}
%
\begingroup
\addtolength{\jot}{-0.5ex}
\begin{align*}
  \ftol{x} & \mapsto  x \\
  \ftol{\neg \varphi} & \mapsto  \neg \ftol{\varphi}\\
  \ftol{\varphi_1 \wedge \varphi_2}  &\mapsto \ftol{\varphi_1} \wedge \ftol{\varphi_2}\\
  \ftol{\varphi_1 \vee \varphi_2} & \mapsto  \ftol{\varphi_1} \vee \ftol{\varphi_2}\\
  \ftol{\LNext \varphi} & \mapsto  \LNext (\ftol{\varphi} \wedge \neg \lend) \\
  \ftol{\LWnext \varphi} & \mapsto  \LNext (\ftol{\varphi} \vee \lend)\\
  \ftol{\LEventually \varphi} & \mapsto  \LEventually (\ftol{\varphi} \wedge \neg \lend)\\
  \ftol{\LGlobally \varphi} & \mapsto  \LGlobally (\ftol{\varphi} \vee \lend)\\
  \ftol{\varphi_1 \LUntil \varphi_2} & \mapsto  \ftol{\varphi_1} \LUntil (\ftol{\varphi_2} \wedge \neg \lend) \\
  \ftol{\varphi_1 \LRelease \varphi_2} & \mapsto  (\ftol{\varphi_1} \wedge \neg \lend) \LRelease (\ftol{\varphi_2} \vee \lend)\\
\end{align*}
\endgroup




\begin{theorem}[\cite{DBLP:conf/aaai/GiacomoMM14}]\label{th:vardidegiacomo}
  Any \LTLf formula without past temporal operators $\varphi$ is
  satisfiable iff the \LTL formula
  \begin{equation}
    \label{eq:ltlf2ltl}
    \LEventually \lend \wedge \LGlobally (\lend \rightarrow \LNext \lend) \wedge \ftol{\varphi}
  \end{equation}
  is satisfiable.
\end{theorem}

Hereafter, we denote with $\ltlftoltl{\varphi}$ the
equation~\eqref{eq:ltlf2ltl} resulting from applying
Theorem~\ref{th:vardidegiacomo}, i.e.,
$\ltlftoltl{\varphi} := \LEventually \lend \wedge \LGlobally (\lend
\rightarrow \LNext \lend) \wedge \ftol{\varphi}$. The resulting \LTL
formula can then be checked for satisfiability with any state of the
art \LTL satisfiability checker as discussed
in~\cite{DBLP:conf/aaai/GiacomoMM14,DBLP:journals/ai/LiPZVR20}.

Finally, there are SAT based frameworks for \LTLf satisfiability
checking like e.g.~\cite{DBLP:journals/ai/LiPZVR20}, where
propositional SAT solving techniques are used to construct a
transition system $T_\varphi$ for a given \LTLf formula $\varphi$, and \LTLf
satisfiability checking is reduced to a path search problem over the
constructed transition system.

\begin{theorem}[\cite{DBLP:journals/ai/LiPZVR20}]\label{th:vardirozier}
  Let $\varphi$ be an \LTLf formula without past temporal
  operators. $\varphi$ is satisfiable iff there is a final state in
  $T_\varphi$.
\end{theorem}
A \emph{final state} for $T_\varphi$ is any state satisfying the Boolean
formula $\lend \wedge (\xnf{\varphi})^p$, where
\begin{enumerate*}[label=(\roman*)]
\item $\lend$ is a new propositional atom to identify the last state
  of satisfying traces (similarly to~\cite{DBLP:conf/aaai/GiacomoMM14});
\item $\xnf{\varphi}$ is the \emph{neXt Normal Form} of $\varphi$, an
  equi-satisfiable formula such that there are no Until/Release
  sub-formulas in the propositional atoms of $\xnf{\varphi}$, built
  linearly from $\varphi$; and
\item $(\xnf{\varphi})^p$ is a propositional formula over the
  propositional atoms of $\xnf{\varphi}$.
\end{enumerate*}
This approach uses a conflict driven algorithm, leveraging on
propositional \mr{unsatisfiable} cores, to perform the explicit path-search.
\begin{marcodo}
  We report hereafter some useful definitions.
\begin{definition}[Conflict Sequence \cite{DBLP:journals/ai/LiPZVR20}]
  Given an \LTLf formula $\varphi$, a \emph{conflict sequence}
  $\mathcal{C}$ for the transition system $T_\varphi$ is a finite
  sequence of set\cdf{s} of states such that:
  \begin{itemize}
  \item The initial state $s_0 = \{\varphi\}$ is in $\mathcal{C}[i]$
    for $0 \le i < |\mathcal{C}|$;
  \item Every state in $\mathcal{C}[0]$ is not a final state;
  \item For every state $s \in \mathcal{C}[i+1]$
    ($0 \le i < |\mathcal{C}|-1$), all the one-transition next
    states of $s$ are included in $\mathcal{C}[i]$.
  \end{itemize}
  We call each $\mathcal{C}[i]$ a \emph{frame}, and $i$ is the
  \emph{frame level}.
\end{definition}
For a given conflict sequence $\mathcal{C}$, the set $\bigcap_{0 \le j < i} \mathcal{C}[j]$ (for $0 \le i < |\mathcal{C}|$) represents a set of states that cannot reach a final state of $T_\varphi$ in up to $i$ steps.
\begin{theorem}[\cite{DBLP:journals/ai/LiPZVR20}]\label{th6:lin}
  The \LTLf formula $\varphi$ is unsatisfiable iff there is a conflict
  sequence $\mathcal{C}$ and an $i \ge 0$ such that $\bigcap_{0 \le j < i} \mathcal{C}[j] \subseteq \mathcal{C}[i+1]$.
\end{theorem}
We refer the reader to~\cite{DBLP:journals/ai/LiPZVR20} for further
details about the construction of $T_\varphi$, for the SAT based
algorithm to check for the existence of a final state in $T_\varphi$,
and for the correctness and termination of such algorithm.
\end{marcodo}

\subsubsection{Symbolic Approaches to Check Language Emptiness for \LTL}
\label{sssec:symbleltl}
The standard symbolic approaches to check language emptiness for a
given \LTL formula $\varphi$~\cite{DBLP:journals/fmsd/ClarkeGH97}
consists in
\begin{enumerate*}[label=(\roman*)]
\item building a Symbolic Non-Deterministic B\"uchi automaton
  for the formula $\varphi$;
\item compute on this automaton the set of
  fair states; and
\item intersect it with the set of initial states.
\end{enumerate*}
The resulting set, denoted with $\den{\varphi}$, is a propositional
formula whose models represent all states that are the initial state
of some infinite trace that accepts $\varphi$.
\begin{marcodo}
  More precisely, let $\mathcal{M_\varphi}$ be a symbolic fair transition system over a set of Boolean variables $\mathcal{V}_\varphi$ that encodes the formula $\varphi$ as discussed for instance in~\cite{DBLP:journals/fmsd/ClarkeGH97}. In this setting, $\mathcal{V}_\varphi = \mathcal{V} \cup \mathcal{V}_{B(\varphi)}$ contains all the propositional variables $\mathcal{V}$ and the Boolean variables $\mathcal{V}_{B(\varphi)}$ (such that $\mathcal{V}_{B(\varphi)} \cap \mathcal{V} = \emptyset$) needed to encode a symbolic fair transition system representing the B\"uchi automaton for $\varphi$.\footnote{\mr{We refer the reader to~\cite{DBLP:journals/fmsd/ClarkeGH97} for
\begin{enumerate*}[label=(\roman*)]
\item the formal definition of symbolic fair transition system and
\item the details on a construction of a symbolic fair transition system $\mathcal{M_\varphi}$ for a given \LTL formula $\varphi$.
\end{enumerate*}}}
Let $\den{\varphi}$ be a set of states of such symbolic fair transition system such that:

\begin{enumerate}[label=(A\arabic*),ref=(A\arabic*),itemsep=0pt,topsep=0pt]
    \item All states in $\den{\varphi}$ are the starting point of some path accepting $\varphi$;\label{ass:ASS1}
    \item All words accepted by $\varphi$ are accepted by some path starting from $\den{\varphi}$.\label{ass:ASS2}
\end{enumerate}
\end{marcodo}

We remark that, this approach is suitable both for BDD based and for SAT based approaches to \LTL satisfiability.

\subsubsection{Temporal Resolution Approaches for \LTL Satisfiability}

\LTL satisfiability can also be addressed with temporal
resolution~\cite{DBLP:conf/ijcai/Fisher91,DBLP:journals/tocl/FisherDP01}.
Temporal resolution extends classical propositional resolution with
specific inference rules for each temporal operator.
Temporal resolution has been implemented in 
solvers like e.g. \trppp~\cite{DBLP:conf/cade/HustadtK03}
showing effectiveness in analyzing unsatisfiable \LTL
formulas~\cite{DBLP:conf/atva/SchuppanD11}.
We refer the reader to~\cite{DBLP:conf/ijcai/Fisher91,
  DBLP:journals/tocl/FisherDP01, DBLP:conf/cade/HustadtK03} for
further details.
We remark that, in \cite{DBLP:journals/acta/Schuppan16}, it was showed
how the temporal resolution proof graph constructed to prove
unsatisfiability of an \LTL formula without past temporal operators
could be used to compute a minimal unsatisfiable core for the
respective \LTL formula.

\section{Extracting unsat cores for \LTLf}
\label{sec:algorithms}
\sloppypar


We present here how four complementary state-of-the-art algorithms can be leveraged to extract unsatisfiable cores for a given set of \LTLf formulas, \cg{following two different pathways. The first pathway comprises algorithms that extend approaches originally developed for \LTL, either relying on satisfiability checking or on temporal resolution; the second pathway instead extends \mr{a reference approach} developed for \LTLf in a native manner.}

\mr{
\subsection{Preliminary results}
}

\mr{This section presents three results that enable the use of the
  different frameworks we will adopt in the two different pathways for
  \LTLf unsat core extraction:
  \begin{enumerate*}[label=(\roman*)]
  \item the extension of the translation function $\ftol{\varphi}$
    presented in Section~\ref{sec:background} to handle \LTLf past
    temporal operators;
  \item a translation that allows to transform any \LTLf formula with
    past temporal operators in an equi-satisfiable one with only
    future temporal operators;
  \item the use of an activation variable associated to each \LTLf
    formula in $\Gamma$ to extract unsatisfiable cores from existing
    frameworks for \LTL/\LTLf satisfiability frameworks.
  \end{enumerate*}
  The first enables the use of any framework for \LTL
  satisfiability checking that supports both past and future temporal
  operators. The second enables the use of any framework for
  \LTL/\LTLf satisfiability checking that supports only future
  temporal operators.  Finally, the latter enables for obtaining the
  unsatisfiable cores of $\Gamma$ leveraging existing \LTL/\LTLf
  satisfiability frameworks by building an equi-satisfiable formula
  with these activation variables and looking at the activation
  variables that will make such equi-satisfiable formula
  unsatisfiable.}

\paragraph{Extending \ftol{} to handle past temporal operators.}
\mr{We make the following observation: the semantics for past temporal
  operators over finite traces coincides with the respective semantics
  on infinite traces (it refers to the prefix of the path). Thus, we
  can extend the $\ftol{\varphi}$ encoding to handle \LTLf past
  temporal operators as follows:}
\begin{displaymath}
 \begin{array}{r@{\;\;\mapsto\;\;}l r@{\;\;\mapsto\;\;}l}
  \ftol{\LPre \varphi} & \LPre (\ftol{\varphi}) &
  \ftol{\LWpre \varphi} & \LWpre (\ftol{\varphi})\\
  \ftol{\LOnce \varphi} & \LOnce (\ftol{\varphi}) &
  \ftol{\LHistorically \varphi} & \LHistorically (\ftol{\varphi})\\[-0.2em]
 \end{array}
\end{displaymath}
\begin{displaymath}
\begin{array}{r@{\;\;\mapsto\;\;}l}
\ftol{\varphi_1 \LSince \varphi_2}   & \ftol{\varphi_1} \LSince \ftol{\varphi_2}\\
\ftol{\varphi_1 \LTrigger \varphi_2} & \ftol{\varphi_1} \LTrigger \ftol{\varphi_2}\\
 \end{array}
\end{displaymath}
Basically, for past temporal operators the encoding is propagated
recursively to the sub-formulas without modifications on the past
operator itself.
This extension together with Theorem~\ref{th:vardidegiacomo} allows us
to prove the following corollary.
\begin{corollary}\label{th:ltlf2ltlwithpast}
  Any \LTLf formula $\varphi$ is satisfiable iff the \LTL formula
  \begin{equation}
    \label{eq:ltlf2ltl2}
    \LEventually \lend \wedge \LGlobally (\lend \rightarrow \LNext \lend) \wedge \ftol{\varphi}
  \end{equation}
  is satisfiable.
\end{corollary}
This corollary enables the use of any framework for \LTL
satisfiability checking that supports both past and future temporal
operators.

\paragraph{Removing past temporal operators.}
Given an \LTLf formula with past operators, we \mr{can} build an
equi-satisfiable \mr{\LTLf} formula over only future operators using
the function $\ltlptoltl{\varphi,\emptyset} = \pair{\varphi',\Upsilon}$ that takes an \LTLf formula with past operators, and builds a new formula
$\varphi'$ and a set of formulas $\Upsilon$ as follows:

\noindent
\begin{align*}
	\ltlptoltl{x,\Upsilon} \mapsto & \pair{x,\Upsilon}\\
  	\ltlptoltl{\sim \varphi,\Upsilon} \mapsto & \pair{\sim \varphi', \Upsilon'} \textrm{ where } 		\pair{\varphi',\Upsilon'} = \ltlptoltl{\varphi, \Upsilon}\\
 		& \textrm{and } \sim \in \{\neg, \LNext, \LWnext, \LEventually, \LGlobally\}\\
  	  \ltlptoltl{\varphi_1 \oplus \varphi_2,\Upsilon} \mapsto & \pair{\varphi_1' \oplus \varphi_2', \Upsilon'} \textrm{ where } \pair{\varphi_1',\Upsilon_1} = \ltlptoltl{\varphi_1, \Upsilon}, \\
   		& \pair{\varphi_2',\Upsilon_2} = \ltlptoltl{\varphi_2, \Upsilon}, \Upsilon' = \Upsilon_1 \cup \Upsilon_2, \textrm{ and}\\
        & \oplus \in \{\wedge, \vee, \LUntil. \LRelease\}\\
	\ltlptoltl{\LWpre \varphi,\Upsilon} \mapsto & \ltlptoltl{\neg \LPre \neg \varphi,\Upsilon}\\
 	\ltlptoltl{\LPre \varphi,\Upsilon} \mapsto & \pair{(\LPre \varphi)^p, \Upsilon''} \textrm{ where } 		\pair{\varphi',\Upsilon'} = \ltlptoltl{\varphi, \Upsilon},\\
        &  \Upsilon''= \Upsilon' \cup \{\neg (\LPre \varphi)^p, \LGlobally (\LNext (\LPre \varphi)^p \leftrightarrow \varphi')\}\\
	\ltlptoltl{\varphi_1 \LSince \varphi_2, \Upsilon} \mapsto & \pair{\varphi_2' \vee (\mr{\varphi_1'} \wedge (\varphi_1 \LSince \varphi_2)^p), \Upsilon'} \textrm{ where } \\
        & \pair{\varphi_1',\Upsilon_1} = \ltlptoltl{\varphi_1, \Upsilon}, \pair{\varphi_2',\Upsilon_2} = \ltlptoltl{\varphi_2, \Upsilon},\\
        & \Upsilon'= \Upsilon_1 \cup \Upsilon_2 \cup \{\neg (\varphi_1 \LSince \varphi_2)^p\} \cup\\
        & \{\LGlobally (\LNext (\varphi_1 \LSince \varphi_2)^p \leftrightarrow (\varphi_2 \vee (\varphi_1 \wedge (\varphi_1 \LSince \varphi_2)^p)))\}\\
 \ltlptoltl{\varphi_1 \LTrigger \varphi_2, \Upsilon} \mapsto &  \mr{\ltlptoltl{\neg(\neg \varphi_1 \LSince \neg \varphi_2), \Upsilon}}
\end{align*}

Intuitively, $\ltlptoltl{\varphi}$ recursively replaces each
sub-formula of $\varphi$ with a past temporal operator with a new
fresh propositional variable, and accumulates in $\Upsilon$ formulas
capturing the \mr{semantics} of the substituted past temporal sub-formulas
(e.g., a kind of monitor). In light of this translation, the
following theorem follows.

\begin{theorem}\label{th:ltlfp2ltlff}
  Any \LTLf formula $\varphi$ is satisfiable if and only if the \LTLf
  formula $\phi' \wedge \bigwedge_{\rho \in \Upsilon}\rho$, where
  $\pair{\phi',\Upsilon} = \ltlptoltl{\varphi,\emptyset}$, is satisfiable.
\end{theorem}
\begin{marcodo}
\begin{proof}
The proof is by cases on the structure of the formula. We consider only the $\LPre$ and $\LSince$ past temporal operators since in all the other cases, the $\ltlptoltl{}$ preserves the formula and/or rewrites it leveraging on equivalences of temporal operators w.r.t. these two past operators.
\begin{itemize}
\item $\LPre \varphi$. \\
  $\Longrightarrow$ Let's assume that there exists a path $\pi$ such that $\pi,i \models \LPre \varphi$ and $i \ge 1$ (i.e. such that $\pi,i-1 \models \varphi$). We can construct a new path $\pi'$ extending the path $\pi$ to consider a new fresh variable $(\LPre \varphi)^p$ such that for $\pi'[0] \models \neg (\LPre \varphi)^p$ and $\forall i \ge 1.\, \pi'[i] \models (\LPre \varphi)^p$ iff $\pi',i-1 \models \varphi$. Thus, $\pi' \models \neg (\LPre \varphi)^p \wedge \LGlobally (\LNext (\LPre \varphi)^p \leftrightarrow \varphi)$, and is such that at $i\ge 1$ it holds $\pi'[i] \models (\LPre \varphi)^p$ by construction, thus $\pi',i \models (\LPre \varphi)^p$.

  $\Longleftarrow$ Let's assume that there exists a path $\pi$ such that $\pi \models \neg (\LPre \varphi)^p \wedge \LGlobally (\LNext (\LPre \varphi)^p \leftrightarrow \varphi)$ and there exists an $i\ge 1$ such that $\pi,i \models (\LPre \varphi)^p$. This path will be such that
    $\pi[0] \models \neg (\LPre \varphi)^p$ and $\forall i.\, i \ge 1. \, \pi[i] \models (\LPre \varphi)^p$ iff $\pi,i-1 \models \varphi$. Thus, $\pi,i \models \LPre \varphi$.

\item $\varphi_1 \LSince \varphi_2$\\
  $\Longrightarrow$ Let's assume that there exists a path $\pi$ such $\pi,i \models \varphi_1 \LSince \varphi_2$. This path is such that $\exists k \textrm{ with } 0 \le k \le i$ \textrm{ such that } $\pi, k \models \varphi_2 \textrm{ and } \forall j \textrm{ with } k < j \le i$ it holds that
$\pi, j \models \varphi_1$. We can build a new path $\pi'$ extending the path $\pi$ to consider a new fresh variable $(\varphi_1 \LSince \varphi_2)^p$ such that $\pi'[0] \models \neg (\varphi_1 \LSince \varphi_2)^p$, and  $\forall i \ge 1. \pi'[i] \models (\varphi_1 \LSince \varphi_2)^p$ iff $\pi',i-1 \models \varphi_2$ or $\pi',i-1 \models \varphi_1 \wedge (\varphi_1 \LSince \varphi_2)^p$. Thus, $\pi' \models \neg (\varphi_1 \LSince \varphi_2)^p \wedge \LGlobally (\LNext (\varphi_1 \LSince \varphi_2)^p \leftrightarrow (\varphi_2 \vee (\varphi_1 \wedge (\varphi_1 \LSince \varphi_2)^p)))$, and it is such that at $i \ge 0$ it holds that $\pi',i \models \varphi_2 $ or  $\pi',i \models \varphi_1 \wedge (\varphi_1 \LSince \varphi_2)^p$ by construction, and thus $\pi',i \models \varphi_2 \vee (\varphi_1 \wedge (\varphi_1 \LSince \varphi_2)^p)$.

  $\Longleftarrow$ Let's assume there is a path $\pi$ such that $\pi \models \neg (\varphi_1 \LSince \varphi_2)^p \wedge \LGlobally (\LNext (\varphi_1 \LSince \varphi_2)^p \leftrightarrow (\varphi_2 \vee (\varphi_1 \wedge (\varphi_1 \LSince \varphi_2)^p)))$ and there exists an $i$ such that $\pi,i \models (\varphi_2 \vee (\varphi_1 \wedge (\varphi_1 \LSince \varphi_2)^p))$. This path will be such that $\exists k \textrm{ with } 0 \le k \le i$ \textrm{ such that } $\pi, k \models \varphi_2 \textrm{ and } \forall j \textrm{ with } k < j \le i$ it holds that $\pi, j \models \varphi_1$, thus $\pi,i \models \varphi_1 \LSince \varphi_2$.
\end{itemize}
\end{proof}
\end{marcodo}

\mr{This result enables the use of any framework for \LTL/\LTLf satisfiability checking that does not support past temporal operators.}

\paragraph{Activation variables.}
To compute the unsatisfiable core for a given set
$\Gamma = \{\varphi_1, ..., \varphi_N\}$ of \LTLf formulas we proceed
as follows. For each \LTLf formula $\varphi_i \in \Gamma$ we introduce
an \emph{activation} variable $A_i$, i.e., a fresh propositional
variable $A_i \not\in \PVarSet$. We then define the \LTLf formula
$\Psi = \bigwedge_{i} (A_i \rightarrow \varphi_i)$.
Let $A = \{A_1, ..., A_N\}$ be the set of activation variables, thus
the formula $\Psi$ is over $\PVarSet \cup A$.

We make the following observation: the satisfiability of $\Gamma$ is
conditioned by the activation variables $A$, and we have the following
theorems.

\begin{theorem}\label{th:unsatwithactivation}
  Let $\Gamma = \{\varphi_1, ..., \varphi_N\}$ be a set of \LTLf
  formulas over $\PVarSet$, $A = \{A_1, ..., A_N\}$ be a set of
  propositional variables such that $A \cap \PVarSet =
  \emptyset$\begin{marcodo}, and $\Psi = \bigwedge_{i} (A_i \rightarrow \varphi_i)$. $\Gamma$ is unsatisfiable if and only if
  $\Psi \wedge \bigwedge_{A_i \in A} A_i$ is unsatisfiable.
   \end{marcodo}
\end{theorem}
\begin{marcodo}
\begin{proof}
$\Longrightarrow$ Let us assume $\Gamma$ unsatisfiable, this means that $\bigwedge_{\varphi_i \in \Gamma} \varphi_i$ is unsatisfiable. Let's now consider $\Psi \wedge \bigwedge_{A_i \in A} A_i$, and let us assume it is satisfiable. This means that there exists a path $\pi$ such that all $A_i \in A$ should be true in the initial state $\pi[0]$, and as a consequence also that all $A_i \rightarrow \varphi_i$ should be satisfiable by such path $\pi$, and also that the conjunction of all $\varphi_i$ should be so. However, this contradicts the hypothesis that $\Gamma$ is unsatisfiable.

\noindent $\Longleftarrow$ Let's assume $\Psi \wedge \bigwedge_{A_i \in A} A_i$ being unsatisfiable.
This means that for all subset $A' \subseteq A$ such that each $A_i \in A'$ is true and all the other variables in $A \setminus A'$ are set to false, the conjunction $\bigwedge_{A_i \in A' }\varphi_i$ is unsatisfiable.
Let's consider $\Gamma$, and let us assume $\Gamma$ is satisfiable. This means that there exists a path $\pi$ such that $\pi,0 \models \bigwedge_{\varphi_i \in \Gamma} \varphi_i$, and also that for all $\varphi_i \in \Gamma$ $\pi,0 \models \varphi_i$. This contradicts the hypothesis that $\Psi \wedge \bigwedge_{A_i \in A} A_i$ \cdf{is} unsatisfiable.
\end{proof}
\end{marcodo}

\begin{theorem}\label{th:unsatcore}
  Let $\UC$ be a subset of $A$. Then the set
  $\Phi_{\UC} = \{\varphi_i | A_i \in \UC\}$ is an unsatisfiable core for $\Psi$
  iff the formula $\Psi \wedge \bigwedge_{A_i \in \UC} A_i$ is
  unsatisfiable.
\end{theorem}
\begin{marcodo}
\begin{proof}
The proof is analogous to the proof of Theorem~\ref{th:unsatwithactivation}.
\end{proof}
\end{marcodo}

This theorem allows us to obtain the unsatisfiable cores (\UCs) of
$\Gamma$ by looking at the activation variables that will make $\Psi$
unsatisfiable.

\subsection{\LTLf Unsatisfiable Core Extraction via Reduction to \LTL}
\label{sec:reduction-to-LTL}

\cg{This section provides details of how we compute \LTLf unsat core extraction via reduction to \LTL satisfiability checking over infinite traces and via \LTL temporal resolution. The first two algorithms we present leverage two different state-of-the-art techniques for \LTL satisfiability checking, namely,
\begin{enumerate*}[label=(\roman*)]
\item Binary Decision Diagrams \mr{(BDDs)}~\cite{DBLP:journals/csur/Bryant92} approaches such as e.g.,~\cite{DBLP:journals/fmsd/ClarkeGH97}; and
\item SAT based approaches such as e.g.,~\cite{DBLP:journals/lmcs/BiereHJLS06}.
\end{enumerate*} The third algorithm instead is based on \mr{temporal resolution for} \LTL~\cite{DBLP:conf/cade/HustadtK03,DBLP:journals/acta/Schuppan16}, extended to support past temporal operators.
}

\mr{ We leverage Theorem~\ref{th:unsatcore} to obtain the
  unsatisfiable cores (\UCs) of $\Gamma = \{\varphi_1,...,\varphi_N\}$
  by looking at the activation variables $A = \{A_1, ..., A_N\}$, with
  $A \cap \PVarSet = \emptyset$, that will make the formula
  $\bigwedge_{i} (A_i \rightarrow \varphi_i) \wedge \bigwedge_{A_i \in
    A} A_i$ unsatisfiable.  }
\mr{In the following,} we show how to obtain \UCs using different
solving techniques.

\subsubsection{BDD based \LTLf Unsatisfiable Core Extraction}
\label{sec:bddmcltlfuc}

Given the set $\Gamma = \{\varphi_1, ..., \varphi_N\}$ of \LTLf
formulas, we build the formula $\Psi$ as discussed in
Theorem~\ref{th:unsatwithactivation}. Then we consider the following
\LTL formula built leveraging Corollary~\ref{th:ltlf2ltlwithpast}:
\begin{equation}
  \label{eq:bddltlf2ltlcore}
  \Psi' = \LEventually \lend \wedge \LGlobally (\lend \rightarrow \LNext \lend) \wedge \ftol{\Psi}
\end{equation}
The set $\den{\Psi'}$ resulting from applying language emptiness
algorithms on $\Psi'$ (i.e. BDDLTLSAT($\Psi'$) in Algorithm~\ref{alg:bddltlalgorithm}) is a propositional formula whose models encode
all states that are the initial state of some infinite trace that
accepts $\Psi'$, and it contains both the activation variables $A$ and
the variables $\PVarSet$ together with the variables
$\PVarSet_{\Psi'}$ needed to encode the symbolic B\"uchi automaton for
$\Psi'$.

\begin{theorem}\label{th:ltlfunsatcore}
  There exists a state $s \in \den{\Psi'}$ and a set $C \subseteq A$
  such that $s \models \bigwedge_{A_i \in C} A_i$ if and only if
  $\bigwedge_{i,A_i \in C} \varphi_i$ is satisfiable.
\end{theorem}
\begin{marcodo}
\begin{proof}
$\Longrightarrow$  Suppose there exists a state $s$ in $\den{\Psi'}$ such that $s \models \bigwedge_{A_i \in C} A_i$. For \ref{ass:ASS1} there exists a path $\pi$ starting from $s$ satisfying $\Psi'$. Since $s \models A_i$ for all $A_i \in C$. Then the path $\pi$ also satisfies $\bigwedge_{i,A_i \in C} \varphi_i$.

\noindent $\Longleftarrow$ Suppose that $\bigwedge_{i,A_i \in C} \varphi_i$ is satisfiable by some word $w$ over the alphabet $2^{\mathcal{V}}$. We extend $w$ to $w'$ such that $w' \models \bigwedge_{A_i\in C} A_i$. Then $w'$ satisfies $\Psi'$. For \ref{ass:ASS2} there exists a path $\pi$ starting from $\den{\Psi'}$ satisfying $w'$. Then $\pi[0]$ satisfies $\bigwedge_{A_i\in C} A_i$.
\end{proof}
\end{marcodo}
This theorem allows for the extraction from $\den{\Psi'}$ of all possible
subsets of the implicants $\varphi_1, ..., \varphi_N$ that are
consistent or inconsistent. In particular, given the set of states
$\den{\Psi'}$ corresponding to the assignments to variables $A$,
$\PVarSet$ and $\PVarSet_{\Psi'}$, the set
$\{s \in 2^{A} | \bigwedge_{i,s \in \den{\Psi'}, s \models A_i}
\varphi_i \textrm{ is unsat}\}$ can be obtained from $\den{\Psi'}$ by
quantifying existentially (projecting) the variables corresponding to
$\PVarSet$ and $\PVarSet_{\Psi'}$ and negating (complementing) the
result.
\begin{equation}
  \label{eq:bddunsatcores}
  \UCset_\Gamma (A) = \neg (\exists \PVarSet. \exists \PVarSet_{\Psi'}. \den{\Psi'})
\end{equation}
$\UCset_\Gamma(A)$ is a propositional formula over variables in $A$ where
each satisfying assignment corresponds to an unsatisfiable core for $\Gamma$.
\begin{marcodo}
\begin{corollary}\label{cor:corollary_uc_bdd}
  $SAT_\Gamma (A) = (\exists \PVarSet. \exists \PVarSet_{\Psi'}. \den{\Psi'}) = \{s \in 2^{A}| \bigwedge_{i, s\models A_i} \varphi_i \text{ is satisfiable} \}$.
  $\UCset_\Gamma (A) = \neg (\exists \PVarSet. \exists \PVarSet_{\Psi'}. \den{\Psi'}) = \{s \in 2^{A}| \bigwedge_{i, s\models A_i} \varphi_i \text{ is unsatisfiable} \}$.
\end{corollary}
\end{marcodo}
Equation~\eqref{eq:bddunsatcores} can be easily implemented with BDDs
through the respective existential quantification and negation BDD
operations~\cite{DBLP:journals/csur/Bryant92,DBLP:conf/cav/CimattiRST07}.


\begin{algorithm}[tb]
\caption{\small BDD \LTL \UC Extraction with~\cite{DBLP:journals/fmsd/ClarkeGH97}}\small
\label{alg:bddltlalgorithm}
\textbf{Input}: $\Psi'$, $A$\\
\textbf{Output}: \mr{$UC$} or $\emptyset$
\begin{algorithmic}[1] 
  \mr{\STATE $\den{\Psi'} \gets$ \Call{BDDLTLSAT}{$\Psi'$}}
  \mr{\STATE $\UCset \gets \neg(\exists \PVarSet. \exists \PVarSet_{\Psi'}. \den{\Psi'})$}
   \mr{\IIF{($\UCset = \emptyset$)} \IRETURN{$\emptyset$} \ENDIIF}
   \mr{\STATE $\UC \gets$ \Call{PickOne}{$\UCset$}}
   \mr{\STATE \IRETURN{$\UC$}}
\end{algorithmic}
\end{algorithm}

Algorithm~\ref{alg:bddltlalgorithm} computes all the unsatisfiable
cores $\UCset$ for a set of \LTLf formulas $\Gamma$ leveraging the BDD
based approach discussed in~\cite{DBLP:journals/fmsd/ClarkeGH97}.
It takes in input a rewritten formula $\Psi'$, and it returns the
empty set ($\emptyset$) if the formula is satisfiable, otherwise
\mr{it returns an $\UC \in \UCset \subseteq 2^A$ such that
  $\forall s \in \UCset. \bigwedge_{i,A_i \in s} \varphi_i$ is
  unsatisfiable}. It uses BDDLTLSAT
algorithm~\cite{DBLP:journals/fmsd/ClarkeGH97} to compute
$\den{\Psi'}$. \mr{See Section~\ref{sssec:symbleltl} and
  \cite{DBLP:journals/fmsd/ClarkeGH97} for a more thorough discussion
  on how the check for language emptiness is performed.}

\begin{theorem}
  Algorithm~\ref{alg:bddltlalgorithm} returns $\emptyset$ if the
  set of \LTLf formulas $\Gamma$ is satisfiable, otherwise it \mr{computes}
  an $\UCset \not=\emptyset$ such that the $\forall \UC \in \UCset$ is
  $\Phi_{\UC} = \{\varphi_i | A_i \in \UC\}$ an unsatisfiable core for
  $\Gamma$\mr{, and then it returns an $\UC \in \UCset$.}
\end{theorem}
\begin{marcodo}
\begin{proof}\sloppypar
The proof is a direct consequence of Theorem~\ref{th:ltlfunsatcore} and Corollary~\ref{cor:corollary_uc_bdd}. Indeed, BDDLTLSAT($\Psi'$) computes $\den{\Psi'}$, i.e. the set of states such that are the starting point of some path satisfying $\Psi'$. If $\Gamma$ is satisfiable, it means that any of its subset will be satisfiable as well, thus any possible assignment to $A_i$ will be such that $\Psi'$ will be satisfiable, and the set $SAT_\Gamma (A) \not= \emptyset$, and thus (line 3 of Algorithm~\ref{alg:bddltlalgorithm}) $\UCset_\Gamma (A) = \emptyset$. On the other hand, if $\Gamma$ is unsatisfiable, Equation~\ref{eq:bddunsatcores} extracts the formula over variables $A_i$ such that each satisfying assignment for such formula corresponds to an unsatisfiable core for $\Psi'$, and this in turn is an unsatisfiable core for $\Gamma$.
\end{proof}
\end{marcodo}

\subsubsection{SAT based \LTLf Unsatisfiable Core Extraction}
\label{sec:satmcltlfuc}

Determining language emptiness of an \LTL formula can also be
performed leveraging any off-the-shelf SAT-based bounded model
checking technique equipped with completeness
check~\cite{DBLP:journals/lmcs/BiereHJLS06,kliveness}.
We observe that, all these approaches can be easily extended to
extract an unsatisfiable core from a conjunction of temporal
constraints leveraging the ability of propositional SAT solvers to
check the satisfiability of a propositional formula $\psi$ under a set
of assumptions specified in form of literals $L = \{l_j\}$,
 i.e., checking the satisfiability of the formula
$\psi' = \psi \wedge \bigwedge_{l_j \in L} l_j$. If $\psi'$ turns out
to be unsatisfiable, then the SAT solver can return a subset
$\UC \subseteq L$ such that $\psi \wedge \bigwedge_{l_j \in \UC} l_j$ is
still unsatisfiable.
SAT-based bounded model checking~\cite{DBLP:journals/ac/BiereCCSZ03}
encodes a finite path of length $k$ with a propositional formula over
the set of variables representing the $\PVarSet$ at each time step
from 0 to $k$. To check for completeness, they typically encode the
fact that the path cannot be extended with states not yet
visited~\cite{DBLP:journals/lmcs/BiereHJLS06}.
\mr{We remark that, in model checking one considers both a transition
  system (i.e. a model) and a temporal logic formula. However, since
  we are concerned on satisfiability of \LTL formulas only, we
  consider it the universal model (i.e. if $\PVarSet$ is the set of
  propositional variables, the initial set of states is $2^\PVarSet$,
  and the transitions relation is equal to
  $2^\PVarSet \times 2^\PVarSet$), which corresponds to
  encode symbolically both the initial set of states, and the transitions
  relation with $\top$.}

The approach proceeds as \mr{illustrated in
  Algorithm~\ref{alg:satbmcltlalgorithm} that takes in input the
  rewritten formula $\Psi'$ and computes an unsatisfiable core for a
  set of \LTLf formulas $\Gamma$ leveraging the \mr{bounded model
    checking} encoding defined
  in~\cite{DBLP:journals/lmcs/BiereHJLS06}.  \mr{It uses a
    \emph{completeness formula} $EncC(\phi,k)$ that is unsatisfiable
    iff $\phi$ is unsatisfiable, and a \emph{witness formula}
    $EncP(\phi,k)$ that is satisfiable iff the \LTL formula $\phi$ is
    satisfiable by a path of length $k$.\footnote{\mr{We refer the reader
      to \cite{DBLP:journals/lmcs/BiereHJLS06} for details on how the
      $EncP(\phi,k)$ and $EncC(\phi,k)$ propositional formulas are
      constructed.}}}  For increasing value of $k$, we} submit the SAT
solver a propositional encoding up to the considered $k$ \mr{of} a path
satisfying the formula $\Psi'$ under the assumption that all the
literals in $A$ are true in the initial time step
($A^{[0]} = \bigwedge_{A_i \in A} A_i^{[0]}$). \mr{This is achieved
  through the call \Call{SAT\_Assume}{EncC($\Psi',k$), $A^{[0]}$} that
  checks the satisfiability of EncC($\Psi',k$) under the assumptions
  \cdf{that} the literals in $A^{[0]}$ are true.}
When \mr{this call} proves the formula unsatisfiable, it is
straightforward to get the corresponding unsatisfiable core from the
SAT solver in terms of a subset of the variables in \mr{$A^{[0]}$, and
  we are done with the search}. \mr{On the other hand, if this call
  returns SAT, we cannot conclude the \LTL formula being
  unsatisfiable. In this case, we need to check if it is satisfiable,
  i.e. if there exists a lasso-shaped path of length $k$ that
  satisfies the propositional formula EncP($\Psi',k$)
  $\wedge \bigwedge_{A_i \in A} A_{i}^{[0]}$. This is achieved with
  the simple call \Call{SAT}{EncP($\Psi',k$)
    $\wedge \bigwedge_{A_i \in A} A_{i}^{[0]}$}. If such call returns
  SAT the \LTL formula is satisfiable, and we are done. Otherwise, we
  increase $k$ and we iterate.}

\mr{The Algorithm~\ref{alg:satbmcltlalgorithm} takes} in input the
rewritten formula $\Psi'$ and returns the empty set ($\emptyset$) if
the formula is satisfiable, otherwise it returns a subset
$\UC \subseteq A$ such that $\bigwedge_{i,A_i \in \UC} \varphi_i$ is
unsatisfiable.
%

\begin{algorithm}[tb]
\caption{\small SAT BMC \LTL \UC Extraction with~\cite{DBLP:journals/lmcs/BiereHJLS06}}\small
\label{alg:satbmcltlalgorithm}
\textbf{Input}: $\Psi'$, $A$\\
\textbf{Output}: $\UC$ or $\emptyset$
\begin{algorithmic}[1] 
  \STATE $k \gets 0$
  \IWHILE{(True)}
   \STATE $res,\UC \gets $ \Call{SAT\_Assume}{EncC($\Psi',k$), $A^{[0]}$}
   \IIF{($res = UNSAT$)} \IRETURN{$\UC$} \ENDIIF
   \STATE $res \gets$ \Call{SAT}{EncP($\Psi',k$) \begin{marcodo}$\wedge \bigwedge_{A_i \in A} A_{i}^{[0]}$\end{marcodo}}

   \IIF{($res = SAT$)} \IRETURN{$\emptyset$} \ENDIIF
   \STATE $k \gets k+1$
  \IENDWHILE
\end{algorithmic}
\end{algorithm}

\begin{marcodo}
\begin{lemma}\label{lemma:satbmc_lemma1}
$EncP(\Psi',k) \wedge \bigwedge_{A_i \in A} A_{i}^{[0]}$ is satisfiable iff there exists a lasso shaped witness $\pi[0,l-1] \pi[l,k]^{\omega}$ such that $\pi \models \bigwedge_{i,A_i \in A} \varphi_i$.
\end{lemma}
\begin{proof}\sloppypar
$\Longrightarrow$ Let $\pi[0,k]$ be a finite path corresponding to a satisfying assignment for $EncP(\Psi',k) \wedge \bigwedge_{A_i \in A} A_{i}^{[0]}$. From assumption~\ref{ass:ASS1} such path should be such that $\exists l$ s.t. $0 \le l \le k \wedge \pi[0,l-1]\pi[l,k]^{\omega} \models \Psi'$.  With $\bigwedge_{A_i \in A} A_{i}^{[0]}$ and the construction of $\Psi'$ we conclude that the projection of $\pi[0,l-1] \pi[l,k]^{\omega}$ onto $\mathcal{V}$ satisfies $\bigwedge_{i,A_i \in A} \varphi_i$.

\noindent $\Longleftarrow$ Let us assume a lasso shaped witness $\pi[0,l-1] \pi[l,k]^{\omega}$ for $\bigwedge_{i,A_i \in A} \varphi_i$. Any extension to $\pi[0,l-1] \pi[l,k]^{\omega}$ such that $\bigwedge_{A_i \in A} A_{i}^{[0]}$ is a witness for $\Psi'$. As a consequence, $EncP(\Psi',k) \wedge \bigwedge_{A_i \in A} A_{i}^{[0]}$ is satisfiable.
\end{proof}

\begin{lemma}\label{lemma:satbmc_lemma2}
Let $\UC\subseteq A$, if $EncC(\Psi',k)$  is unsatisfiable under assumption $\bigwedge_{A_i\in \UC} A_i^{[0]}$, then $\bigwedge_{i,A_i \in \UC} \varphi_i$ is unsatisfiable.
\end{lemma}
\begin{proof}
Let $EncC(\Psi',k)$  be unsatisfiable under assumption $\bigwedge_{A_i\in \UC} A_i^{[0]}$. Let's assume there exists a witness $\pi$ of $\bigwedge_{i,A_i \in \UC} \varphi_i$. We can extend such path $\pi$ to a path $\pi'$ over variables $\mathcal{V} \cup A$ such that $\bigwedge_{A_i \in \UC} A_i^{[0]}$ is a witness for $\Psi'$, and this contradicts the  assumption~\ref{ass:ASS2}.
\end{proof}
\end{marcodo}

\begin{theorem}
  Algorithm~\ref{alg:satbmcltlalgorithm} returns $\emptyset$ if
  the set of \LTLf formulas $\Gamma$ is satisfiable, otherwise it
  returns an $\UC \not=\emptyset$ such that the set
  $\Phi_{\UC} = \{\varphi_i | A_i \in \UC\}$ is an unsatisfiable core
  for $\Gamma$.
\end{theorem}

\begin{marcodo}
\begin{proof}
The proof is a direct consequence of Lemma~\ref{lemma:satbmc_lemma1} and Lemma~\ref{lemma:satbmc_lemma2}.
\end{proof}
\end{marcodo}
The algorithm above uses the
\cite{DBLP:journals/lmcs/BiereHJLS06}
encoding \mr{for both $EncC(\Psi',k)$ and $EncP(\Psi',k)$}. We remark that the schema can also be easily adapted to
leverage \mr{other} algorithms \mr{e.g.} based on $k$-liveness~\cite{kliveness} or
liveness to safety~\cite{liveness2safety} both relying on the
IC3~\cite{bradley} algorithm. What intuitively changes is the
propositional encoding of the \LTL formula and the calls to the SAT
solver to reflect the IC3 algorithm.

\subsubsection{\cg{Temporal Resolution based \LTLf Unsatisfiable Core Extraction}}
\label{sec:resolution}
We can extract the unsat core of a set of \LTLf formulas
$\Gamma = \{\varphi_1, ..., \varphi_N\}$ via \LTL temporal
resolution~\cite{DBLP:conf/cade/HustadtK03} (TR) leveraging the results
previously discussed in this paper and existing \LTL temporal
resolution engines equipped for temporal unsat core
extraction~\cite{DBLP:journals/acta/Schuppan16}.
The approach is as follows.
First we build the formula
$\Psi = \bigwedge_i (A_i \rightarrow \varphi_i)$.
Second, we apply Theorem~\ref{th:ltlfp2ltlff} to remove the past
temporal operators.
Third, we leverage  Theorem~\ref{th:vardidegiacomo} to convert the
\LTLf formula into an equi-satisfiable \LTL formula.
Finally, the resulting \LTL formula is \mr{given in input} to any \LTL
temporal resolution solver suitable to extract a temporal
unsatisfiable core
(e.g. \trppp~\cite{DBLP:journals/acta/Schuppan16}) by
enforcing the activation variables to hold \mr{(i.e. enforcing
  $\bigwedge_{A_i\in A} A_i$ in conjunction with the resulting \LTL
  formula)}.
If the \LTL temporal resolution solver responds UNSAT, looking at the
activation variables in the extracted temporal unsatisfiable core we
get an unsat core of the original set of \LTLf formulas.

\begin{algorithm}[tb]
\caption{\small TR \LTL \UC Extraction with~\cite{DBLP:journals/acta/Schuppan16}}\small
\label{alg:trltlalgorithm}
\textbf{Input}: $\Psi'$, $A$\\
\textbf{Output}: $\UC$ or $\emptyset$
\begin{algorithmic}[1] 
  \STATE $\pair{\phi',\Gamma} \gets \ltlptoltl{\Psi',\emptyset}$
  \STATE $\psi \gets \phi' \wedge \bigwedge_{\varphi \in \Gamma}\varphi \wedge \bigwedge_{A_i \in A} A_i$
  \STATE $res,\UC \gets$ \Call{\trppp}{$\psi$}
  \IIF{$(res = UNSAT)$} \IRETURN{$\UC|_A$}
  \RETURN{$\emptyset$}
\end{algorithmic}
\end{algorithm}

Algorithm~\ref{alg:trltlalgorithm} computes an unsatisfiable core for
a set of \LTLf formulas $\Gamma$ leveraging a \LTL temporal resolution
prover equipped for extracting an unsatisfiable core like, e.g.,
TRP++~\cite{DBLP:journals/acta/Schuppan16}. It takes in input the
formula $\Psi'$ and returns the empty set ($\emptyset$) if the formula
is satisfiable, otherwise it returns a subset $\UC \subseteq A$ such
that $\bigwedge_{i,A_i \in \UC} \varphi_i$ is unsatisfiable. It uses
the TRP++ algorithm~\cite{DBLP:journals/acta/Schuppan16} to compute
the unsatisfiable core $\UC$, and then it extracts from this only the
formulas corresponding to $A_i \in A$ (denoted in the algorithm with
$\UC|_A$). \mr{The TRP++($\psi$) algorithm discussed in
  \cite{DBLP:journals/acta/Schuppan16} first converts the \LTL formula
  $\psi$ into an equi-satisfiable set of Separated Normal Form
  (SNF)~\cite{DBLP:conf/ijcai/Fisher91} clauses $C$, and then it
  checks whether this set is satisfiable or not, and in case of
  unsatisfiability it computes an unsatisfiable core
  $C^{uc} \subseteq C$ and returns it by applying a reconstruction w.r.t. the original set of \LTL formulas. We refer the reader to
  \cite{DBLP:journals/acta/Schuppan16} for the details of the
  algorithm and for the respective proof of correctness of the TRP++
  algorithm.}

\begin{theorem}
  Algorithm~\ref{alg:trltlalgorithm} returns $\emptyset$ if the
  set of \LTLf formulas $\Gamma$ is satisfiable, otherwise it returns
  an $\UC \not=\emptyset$ such that the set
  $\Phi_{\UC}\! =\! \{\varphi_i | A_i\! \in\! \UC\}$ is an unsatisfiable core
  for $\Gamma$.
\end{theorem}
\mr{%
\begin{proof}
  If the set $\Gamma$ is satisfiable, then the formula
  $\psi = \phi' \wedge \bigwedge_{\varphi \in \Gamma}\varphi \wedge
  \bigwedge_{A_i \in A} A_i $ where
  $\pair{\phi',\Gamma} = \ltlptoltl{\Psi',\emptyset}$, is satisfiable
  since it leverages on transformations that preserve satisfiability
  (as proved in Theorems~\ref{th:ltlfp2ltlff},\ref{th:vardidegiacomo}
  and \ref{th:unsatcore}). Thus also \Call{TRP++}{$\psi$} returns sat,
  and the algorithm returns $\emptyset$ to indicate the formula is
  satisfiable.

  On the other hand, if $\Gamma$ is unsatisfiable, then also
  $\psi = \phi' \wedge \bigwedge_{\varphi \in \Gamma}\varphi \wedge
  \bigwedge_{A_i \in A} A_i $ is unsatisfiable. Thus
  \Call{TRP++}{$\psi$} returns UNSAT, together with an unsatisfiable
  core for the formula $\psi$. We remark that, given the structure of
  $\psi$, each $A_i$ will be then converted into an SNF clause
  $c_{A_i} = A_i$ which thus will be part of the set of SNF clauses
  $C$ used internally by TRP++. Since this formula is unsatisfiable,
  the TRP++ algorithm will extract an unsat core
  $\UC = C^{uc} \subseteq C$ such that $C^{uc}$ is
  unsatisfiable. $C^{uc}$ among other clauses will contain some
  $c_{A_i}$ for $A_i \in A$ that will correspond to the respective
  formulas in $\Gamma$ (thanks also to Theorem~\ref{th:unsatcore}),
  and thus this set will represent an unsatisfiable core for $\Gamma$.
\end{proof}
}

\subsection{\LTLf Unsatisfiable Core Extraction via Native SAT}
\label{sec:native-sat}

We adapted the native SAT based \LTLf satisfiability approach
discussed in~\cite{DBLP:journals/ai/LiPZVR20} to extract the
unsatisfiable core. \mr{Since the original approach for \LTLf
  satisfiability checking in~\cite{DBLP:journals/ai/LiPZVR20} was not
  supporting past temporal operators, we rely on
  Theorem~\ref{th:ltlfp2ltlff} to get rid of the past temporal
  operators obtaining an equi-satisfiable \LTLf formula without past
  temporal operators}.

The algorithm \mr{originally} discussed
in~\cite{DBLP:journals/ai/LiPZVR20} can be extended to extract
the unsat core of a set $\Gamma = \{\varphi_1, ..., \varphi_N\}$ of
\LTLf formulas as follows (see
Algorithm~\ref{alg:satltlfalgorithm}). First, we build the \LTLf
formula $\Psi = \bigwedge_i (A_i \rightarrow \varphi_i)$. Then, we
apply Theorem~\ref{th:ltlfp2ltlff} to get an equi-satisfiable
formula. The resulting formula (without past temporal operators) is
then passed in input to the algorithm SATLTLF discussed
in~\cite{DBLP:journals/ai/LiPZVR20}. \mr{The SATLTLF algorithm in~\cite{DBLP:journals/ai/LiPZVR20} has been modified to enforce that each call} to the SAT solver \mr{under assumption}
performed in the algorithm
\mr{assumes also that} the activation
variables $A$ are all true.
\mr{We refer the reader to~\cite{DBLP:journals/ai/LiPZVR20} for a
  thorough description of the algorithm (which is out of the scope of
  this paper). We remark that, the only modifications performed to the
  algorithm consists in adding the assumptions on the activation
  variables to the already considered assumptions in each call to
  satisfiability under assumption already performed within the
  algorithm itself.
  Intuitively, the algorithm in~\cite{DBLP:journals/ai/LiPZVR20}
  constructs a \emph{conflict sequence}
  $\mathcal{C} = \mathcal{C}[0], ..., \mathcal{C}[k]$ (i.e., sequences
  of states that cannot reach a final state of the transition system
  $T_\phi$ constructed from the formula $\phi$ given in input)
  extracted from the unsat cores resulting from different
  propositional unsatisfiable queries performed in the algorithm
  itself. As per Theorem~\ref{th6:lin} the input \LTLf formula is
  unsatisfiable iff there exists a conflict sequence $\mathcal{C}$ and
  an integer $i \ge 0$ such that
  $\bigcap_{0\le j \le i} \mathcal{C}[j] \subseteq \mathcal{C}[i+1]$.}
When the SATLTLF algorithm returns UNSAT, we extract from the last
element of the \mr{computed conflict sequence
  (i.e. $\mathcal{C}[i+1]$)} the unsat core $\UC \subseteq A$, and the
set $\Gamma' = \{\varphi_i | A_i \in \UC\}$ is an unsat core for
$\Gamma$ leveraging on Theorem~\ref{th6:lin}.

\begin{algorithm}[tb]\small
\caption{\small SAT \LTLf \UC Extraction with~\cite{DBLP:journals/ai/LiPZVR20}}
\label{alg:satltlfalgorithm}
\textbf{Input}: $\Psi$, $A$\\
\textbf{Output}: $\UC$ or $\emptyset$
\begin{algorithmic}[1] 
  \STATE $\pair{\phi',\Upsilon} \gets \ltlptoltl{\Psi,\emptyset}$
  \STATE $\psi \gets \phi' \wedge \bigwedge_{\rho \in \Upsilon}\varphi$
  \STATE $res \gets $ SATLTLF($\psi$, $A$)
  \IIF{($res = UNSAT$)}
  \IRETURN{SAT\textsubscript{SOLVER}.GET\_\UC($A$)} \ENDIIF
  \RETURN $\emptyset$
\end{algorithmic}
\end{algorithm}

\begin{theorem}
  Algorithm~\ref{alg:satltlfalgorithm} returns $\emptyset$ if the
  set of \LTLf formulas $\Gamma$ is satisfiable, otherwise it returns
  an $\UC \not=\emptyset$ such that the set
  $\Phi_{\UC} = \{\varphi_i | A_i \in \UC\}$ is an unsatisfiable core
  for $\Gamma$.
\end{theorem}
\mr{%
\begin{proof}
  If the set $\Gamma$ is satisfiable, then the formula
  $\psi = \phi' \wedge \bigwedge_{\rho \in \Upsilon}\varphi$, where
  $\pair{\phi',\Upsilon} = \ltlptoltl{\Psi,\emptyset}$, is also
  satisfiable since it leverages on transformations that preserve
  satisfiability (as proved in
  Theorems~\ref{th:ltlfp2ltlff},\ref{th:vardidegiacomo} and
  \ref{th:unsatcore}). Thus leveraging the correctness and
  completeness of the SATLTLF~\cite{DBLP:journals/ai/LiPZVR20}, then
  SATLTLF will report sat, and in turn the
  Algorithm~\ref{alg:satltlfalgorithm} will report $\emptyset$ to
  indicate the set $\Gamma$ is satisfiable.

  On the other hand, if $\Gamma$ is unsatisfiable, then also the
  formula $\psi = \phi' \wedge \bigwedge_{\rho \in \Upsilon}\varphi$
  will be unsatisfiable. From Theorem~\ref{th6:lin}, then the SATLTLF
  algorithm will build a conflict sequence $\mathcal{C}$ and there
  exists an integer $i \ge 0$ such that
  $\bigcap_{0\le j \le i} \mathcal{C}[j] \subseteq \mathcal{C}[i+1]$,
  $\mathcal{C}[i+1]$ will represent the states such that from there it
  is not possible to reach a final state for $T_\psi$, i.e. there is
  no path starting from these states that will satisfy $\psi$,
  i.e. all the paths from this states will not satisfy $\psi$.

  Thus, for all state $s \in \mathcal{C}[i+1]$ there exists a set
  $C \subseteq A$ such that $s \models \bigwedge_{A_i \in C} A_i$, and
  the formula $\bigwedge_{i,A_i \in C} \varphi_i$ will be
  unsatisfiable, and the set $C$ corresponds to an unsat core
  extracted from $\mathcal{C}[i+1]$ by construction of the SATLTLF
  algorithm (see \cite{DBLP:journals/ai/LiPZVR20} for further
  details).
\end{proof}
}

\mr{
\subsection{Discussion}

We observe that all the described approaches extract one unsat core,
though not necessarily a minimum/minimal one.
\mr{Algorithm~\ref{alg:bddltlalgorithm} could also be easily extended
  to get the minimal $\UC$ from the $\UCset$ set of all possible
  unsatisfiable cores for the given formula.}
\mr{For the SAT based approaches, } a minimum/minimal unsat core could
be extracted by leveraging the ability of the SAT solver to get a
minimum/minimal propositional unsat core. \mr{Similarly, the temporal
resolution solver could be instrumented } to get a minimum/minimal
core.  In all cases, it might be possible to get a minimum/minimal one
with specialized solvers and/or with additional search.
However, this is left for future work.
}




\section{Experimental Evaluation}
\label{sec:evaluation}

\begin{claudiodo}
In this section, we provide details on the implementations of the proposed algorithms \cdf{(Section~\ref{sub:implementation})}, and then we describe the setup and the data sets used for the experimental evaluation \cdf{(Section~\ref{sub:the_experimental_setup})}. We conclude with a report on the results alongside an examination thereof \cdf{(Section~\ref{sub:results})}. 
\end{claudiodo}
%

\subsection{Implementation of the Algorithms} 
\label{sub:implementation}

We implemented Algorithms~\ref{alg:bddltlalgorithm}~and~\ref{alg:satbmcltlalgorithm} as extensions of the \nusmv model
checker~\cite{DBLP:conf/cav/CimattiCGGPRST02} exploiting
the built-in support for past temporal operators,
the $\ftol{\varphi}$ conversion, and Eq.~\eqref{eq:ltlf2ltl2}.
In particular, we enhanced
\begin{enumerate*}[(i)]
\item the BDD-based algorithm for \LTL language
  emptiness~\cite{DBLP:journals/fmsd/ClarkeGH97} and
\item the SAT-based approaches~\cite{DBLP:journals/lmcs/BiereHJLS06}.
\end{enumerate*}
We shall henceforth refer to these tools as \nusmv-B and \nusmv-S, respectively.
\mr{The source code for the extended version of \nusmv with these implementations is available at \url{https://github.com/roveri-marco/ltlfuc}.}

We implemented Algorithm~\ref{alg:satltlfalgorithm} within an extended version of the
\aaltaf tool~\cite{DBLP:journals/ai/LiPZVR20},
with a novel dedicated extension to
support past temporal operators through
$\ltlptoltl{\varphi,\emptyset}$.
\mr{The source code for our extended version of \aaltaf is available at \url{https://github.com/roveri-marco/aaltaf-uc}.}

%
%
We implemented Algorithm~\ref{alg:trltlalgorithm} as a toolchain. First, it
calls our variant of \aaltaf to generate a file that is suitable for
the \trppp temporal resolution
solver~\cite{DBLP:conf/cade/HustadtK03} using the $\ftol{\varphi}$
conversion as per Eq.~\eqref{eq:ltlf2ltl2} and
$\ltlptoltl{\varphi,\emptyset}$. Then, the resulting file is submitted to
\trppp, and finally the generated \UC is post-processed to extract
the auxiliary variables $A$. \mr{For the experiments, we used the latest version of \trppp downloadable from \url{http://www.schuppan.de/viktor/trp++uc/}.}

\subsection{The Experimental Setup} 
\label{sub:the_experimental_setup}

\begin{table}
	\caption{Benchmarks}
	\label{tab:families}
	\begingroup
	\renewcommand*{\arraystretch}{1.4}
	\begin{adjustbox}{max width=\textwidth}

\begin{tabular}{@{}l S[table-format=2.0]S[table-format=2.0]S[table-format=2.0]S[table-format=3.2] >{\hspace{2em}} l S[table-format=2.0]S[table-format=2.0]S[table-format=2.0]S[table-format=3.2]@{}}
	\toprule
	                                      &                   &     \multicolumn{3}{c}{\textbf{Clauses}}      &                                                          &                   &     \multicolumn{3}{c}{\textbf{Clauses}}      \\ \cmidrule{3-5}\cmidrule{8-10}
	\textbf{Family}                       & \textbf{Problems} & \textbf{Min.} & \textbf{Max.} & \textbf{Avg.} & \textbf{Family}                                          & \textbf{Problems} & \textbf{Min.} & \textbf{Max.} & \textbf{Avg.} \\ \midrule
	LTL-as-LTLf/acacia/demo-v3/demo-v3:cl & 11                & 9             & 49            & 29.00         & LTL-as-LTLf/anzu/genbuf/genbuf:cl                        & 20                & 58            & 461           & 231.50        \\
	LTL-as-LTLf/alaska/lift/lift          & 17                & 13            & 29            & 21.00         & LTL-as-LTLf/forobots                                     & 38                & 6             & 6             & 6.00          \\
	LTL-as-LTLf/alaska/lift/lift:b        & 17                & 12            & 32            & 22.47         & LTL-as-LTLf/rozier/counter/counter                       & 19                & 6             & 24            & 15.00         \\
	LTL-as-LTLf/alaska/lift/lift:b:f      & 17                & 12            & 32            & 22.47         & LTL-as-LTLf/rozier/counter/counterCarry                  & 19                & 8             & 26            & 17.00         \\
	LTL-as-LTLf/alaska/lift/lift:b:f:l    & 17                & 14            & 50            & 32.47         & LTL-as-LTLf/rozier/counter/counterCarryLinear            & 19                & 8             & 8             & 8.00          \\
	LTL-as-LTLf/alaska/lift/lift:b:l      & 17                & 14            & 50            & 32.47         & LTL-as-LTLf/rozier/counter/counterLinear                 & 18                & 6             & 6             & 6.00          \\
	LTL-as-LTLf/alaska/lift/lift:f        & 17                & 13            & 29            & 21.00         & LTL-as-LTLf/rozier/formulas/n                            & 30                & 1             & 4             & 1.33          \\
	LTL-as-LTLf/alaska/lift/lift:f:l      & 17                & 15            & 47            & 31.00         & LTL-as-LTLf/schuppan/O1formula                           & 27                & 4             & 1002          & 224.00        \\
	LTL-as-LTLf/alaska/lift/lift:l        & 17                & 15            & 47            & 31.00         & LTL-as-LTLf/schuppan/O2formula                           & 27                & 2             & 1000          & 222.00        \\
	LTL-as-LTLf/anzu/amba/amba:c          & 17                & 75            & 351           & 213.47        & LTL-as-LTLf/schuppan/phltl                               & 13                & 5             & 101           & 30.85         \\
	LTL-as-LTLf/anzu/amba/amba:cl         & 17                & 77            & 369           & 223.47        & LTLf-specific/benchmarks:ltlf/LTLfRandomConjunction/C100 & 500               & 118           & 154           & 131.50        \\
	LTL-as-LTLf/anzu/genbuf/genbuf        & 20                & 1             & 1             & 1.00          & LTLf-specific/benchmarks:ltlf/LTLfRandomConjunction/V20  & 425               & 13            & 146           & 81.54         \\ \cline{7-10}
	LTL-as-LTLf/anzu/genbuf/genbuf:c      & 20                & 57            & 441           & 221.00        & \textbf{Overall}                                         & 1377              & 1             & 1002          & 97.63         \\ \bottomrule
\end{tabular}
	\end{adjustbox}
	\endgroup
\end{table}

For the experimental evaluation, we considered all the unsatisfiable
problems reported 
in~\cite{DBLP:journals/ai/LiPZVR20}, for a total of
\cdc{\num{1377}}
problems.
\cdc{To select the specifications of interest to our analysis out of the original testbed, we included only those for which at least one solver declared the set was unsatisfiable and no other tool contradicted the result, as per the experimental data reported by Li et al.}
\cdc{To compute the $\Gamma$ set}, we considered all the top-level conjuncts of each formula \mr{in the benchmark set}. \mr{For every benchmark, we used the variant of the formula in the \aaltaf format as an input.} \cdf{For the other tools \cdc{but \aaltaf,} we implemented \mr{within} \aaltaf \mr{dedicated converters} to the \mr{respective tool} input format.}

We carried out the experimental evaluation considering the four implementations provided by \nusmv-B, \nusmv-S, our variant of \aaltaf, and the \trppp toolchain.
We ran all experiments on an Ubuntu 18.04.5 LTS
machine, 
8-Core Intel(\textregistered) Xeon(\textregistered) at \SI{2.2}{\GHz},
equipped with \SI{64}{GB} of RAM. \mr{We set a memory occupation limit of \SI{4}{GB}, and a CPU usage limit of \SI{10}{\minute}.}
Additionally, we considered  $k=50$ as the maximum depth for \nusmv-S, and we ran \nusmv-B with the BDD dynamic variable reordering \cdc{mode} active~\cite{DBLP:conf/eurodac/FeltYBS93} to \cdc{dynamically reduce} the size of the BDDs \cdf{and thus} \cdc{save} space over time.%
\footnote{\mr{These \cdc{settings} are motivated by similar choices performed in the experimental evaluations carried out in~\cite{DBLP:conf/cav/CimattiRST07,DBLP:journals/acta/Schuppan16}.}}
\cdc{Whenever the wall-clock timing reported by the implemented technique fell under the lowest sensitivity of the tool (thus being represented as \num{3500565600.000}), we replaced the timing with the minimum non-zero timing reported overall ($3.78\times10^{-4}$).}

\cdc{Finally, we have categorized the benchmarking specifications into \num{25} families, according to their characteristics and provenance. Table~\ref{tab:families} shows the number of specifications per family, along with the minimum, maximum and average number of clauses within.}
\mr{%
In particular, the \textsl{LTLf-specific/benchmarks/LTLFRandomConjunction/V20} and
\textsl{LTLf-specific/benchmarks/LTLFRandomConjunction/C100}
benchmarks are conjunctions of formulas, each selected randomly from
standard patterns. They are characterized by a \cdc{\emph{temporal depth} (i.e., the maximum 
	nesting of temporal operators)} 
of up to \num{3}, with \num{20} propositional
variables.
The number of conjunctions ranges from \num{20} to \num{100} for the
former, 
and from \num{10} to \num{100} for the latter.
The \textsl{LTL-as-LTLf/rozier/counter/*} benchmarks are characterized
by having a small number of propositional variables, and temporal
formulas of different temporal depths (from \num{2} to \num{20}).
The benchmarks in \textsl{LTL-as-LTLf/schuppan/O1Formula} are
characterized by a large number of propositional variables (from \num{1} to
\num{1000}) with temporal formulas of small depths (\num{2} to \num{3}) and different
operators. The benchmarks in \textsl{LTL-as-LTLf/schuppan/O2Formula}
are big conjunctions of formulas of the form
$\LGlobally \LEventually a_i \leftrightarrow a_j$ with
$a_i \not= a_j$.
}

\mr{All the material to reproduce the experiments reported hereafter
  is available at \url{https://github.com/roveri-marco/ltlfuc/archive/refs/tags/jair-release-v0.zip}}.


\subsection{\cg{The Results}}
\label{sub:results}

\mr{In the experimental evaluation, we considered the following evaluation metrics:
\begin{enumerate*}[label=(\roman*)]
\item the result of the check (expecting all the tools to return unsatisfiability and extract an unsatisfiable core if no resource limit is reached);
\item the search time to compute and return an unsatisfiable core;
\item the size of the computed unsatisfiable core.
\end{enumerate*}}
\cdf{We remark that none of the presented algorithms strives for finding a minimum unsatisfiable core.}
\mr{\cdc{The approach based on \trppp may be used to that end, and} \nusmv-B could in principle be easily
adapted to select from the intermediate computed set $\UCset$ a minimum unsatisfiable core (as discussed previously). 
\cdc{Nevertheless, we pick the first returned $\UC$ for all tools so as to have a fair comparison among them.}
}

\begin{figure}[tbp!]
  \centering
  \begin{tabular}{@{}cc@{}}
  \includegraphics[width=0.45\textwidth]{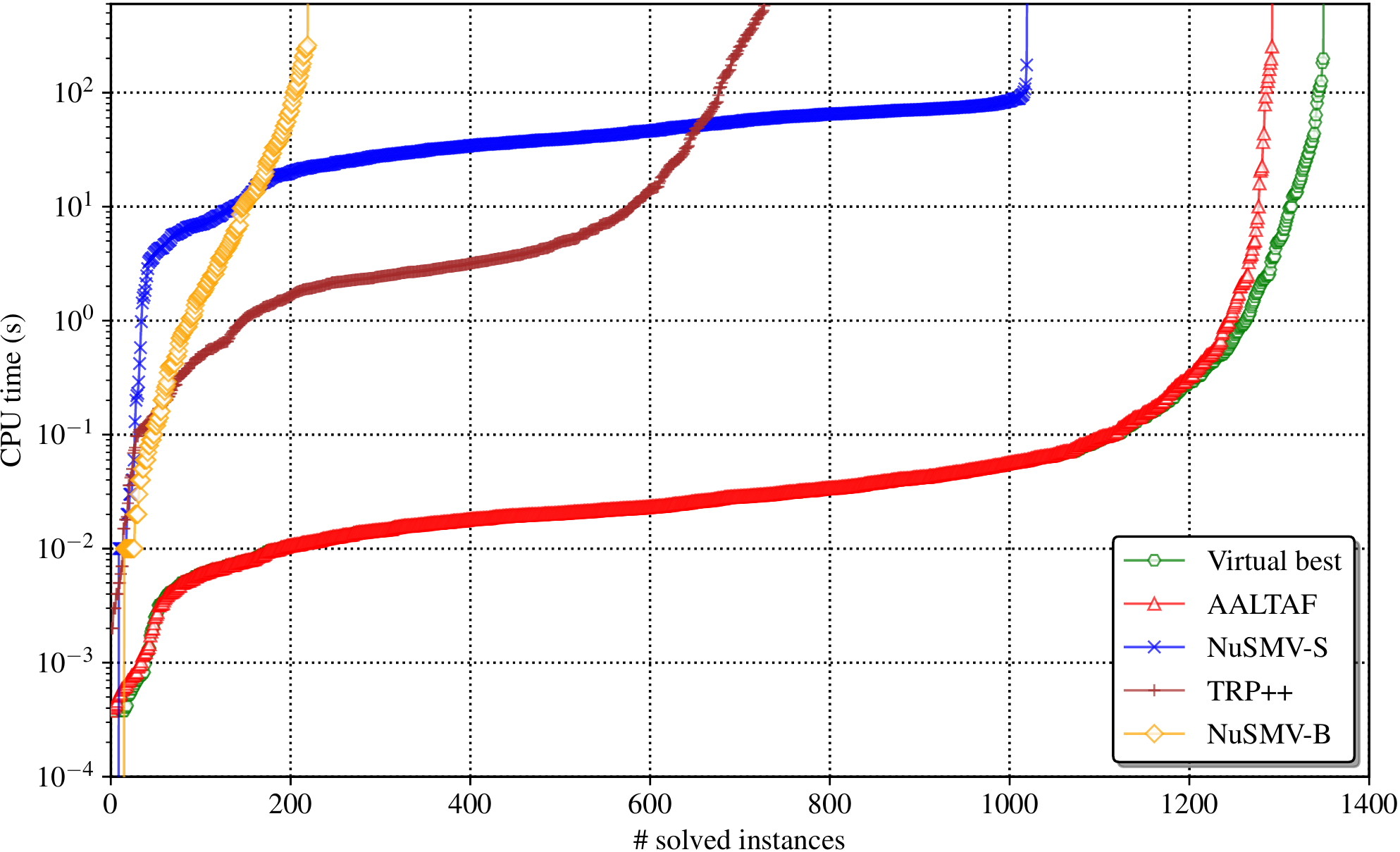} &
  \includegraphics[width=0.45\textwidth]{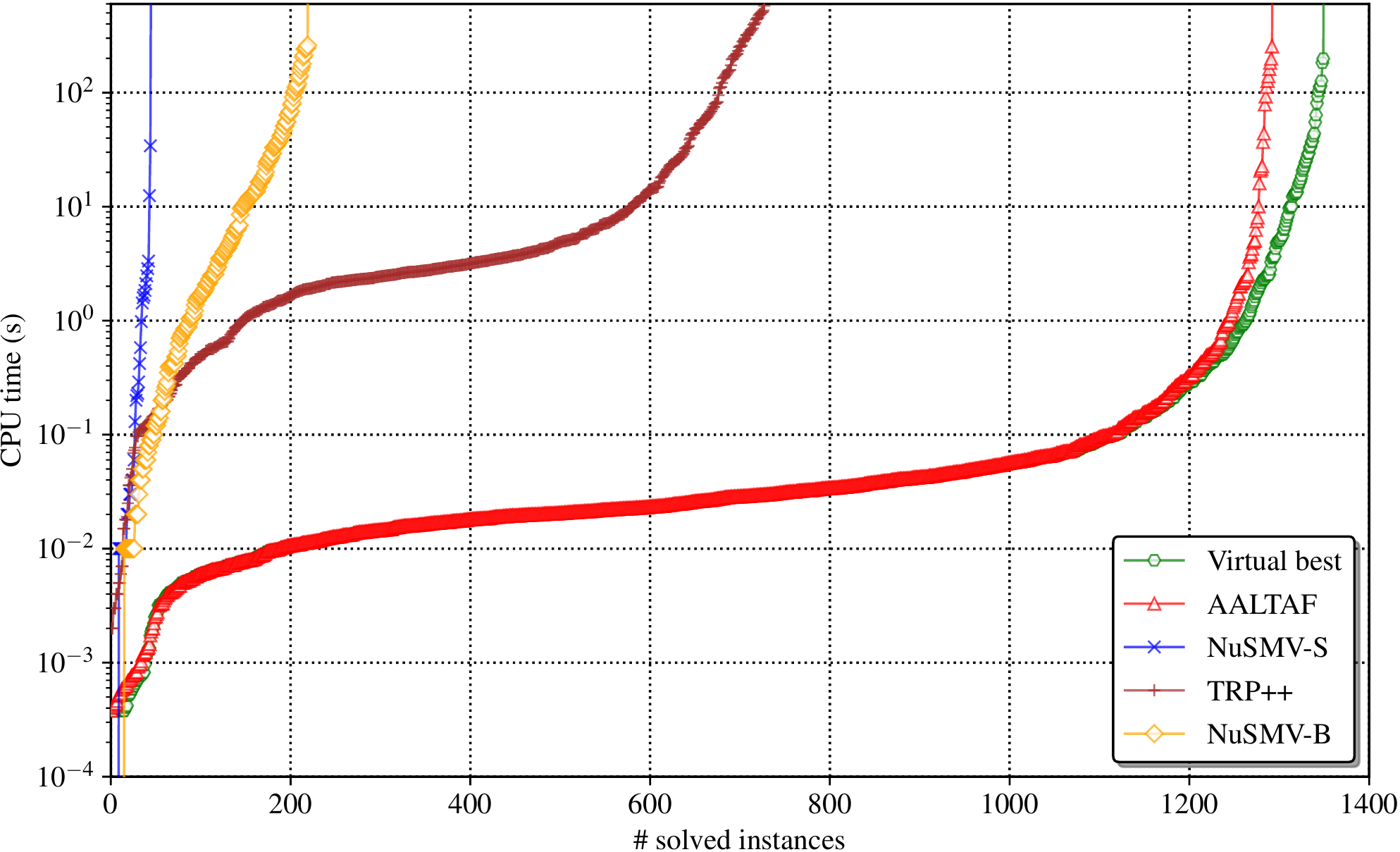}\\
    (a) & (b)
  \end{tabular}
  \caption{\mr{Cactus plots for the whole experimental evaluation.}}
  \label{fig:cactus}
\end{figure}


\begin{marcodo}
  The first result \cdc{is that, as expected,} all the tools reported consistent output
  when terminating without reaching a resource limit (being it memory,
  time or \cdc{search-space} depth). In other words, for all the considered benchmarks it was
  never the case that an algorithm 
  \cdc{declared the specification as satisfiable. This outcome is in line with the original findings in~\cite{DBLP:journals/ai/LiPZVR20}}. 
  \cdc{However, we remark} that 
  individual algorithms could extract different unsatisfiable cores 
  \cdc{among the diverse} possible ones.

  \cdf{Concerning the computation time,} Fig.~\ref{fig:cactus}(a) shows the number of problems solved by each
  algorithm
	\cdf{by \cdc{the \SI{10}{\minute} timeout on the abscissae and the time taken to solve them cumulatively on the ordinates.} 
	\cdc{Alongside the aforementioned tools, the figure illustrates the performance of the} 
	\emph{virtual best}, 
	that is the 
	minimum time required for each solved instance among the four implementations.}
  %
  Fig.~\ref{fig:cactus}(b) is similar to Fig.~\ref{fig:cactus}(a), although
  it 
  excludes the timings for \cdf{the} \nusmv-S runs that reached $k=50$
  albeit being unable to prove unsatisfiability (\emph{unknown}
  answer) and thus construct an unsatisfiable core.
  In this figure, \cdc{then,} the virtual best considers only the cases where the tool computed an unsatisfiable core. \cdc{Both plots show that \aaltaf outperforms the other implementations in the majority of cases, although the tail of the virtual-best curve on the right-hand side \cdf{of both plots} exhibits an influence from \trppp and \nusmv-B\mr{, thus witnessing complementarity of the proposed approaches}. In the remainder of this section, we investigate the comparative assessment more in depth.} 
  \cdc{The overall minimum, maximum, average, and median best timings to return an $\UC$
  	are \SIlist{0.0004;198.5054;1.4931;0.0282}{\second}, respectively.}

\begin{figure}[tbp]
  \centering
  \begin{tabular}{@{}c@{}c@{}}
    \includegraphics[width=0.48\textwidth]{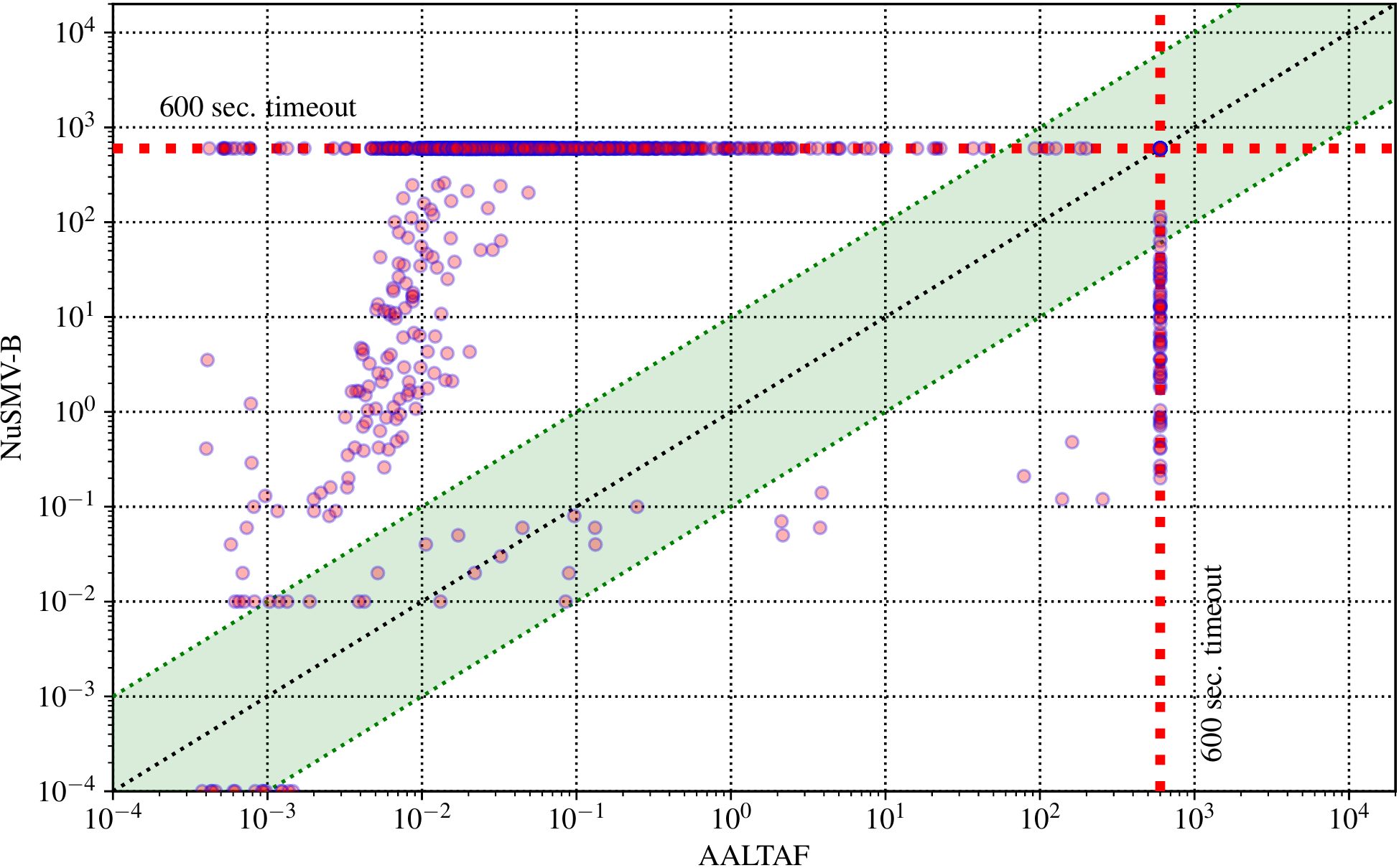} &~
    \includegraphics[width=0.48\textwidth]{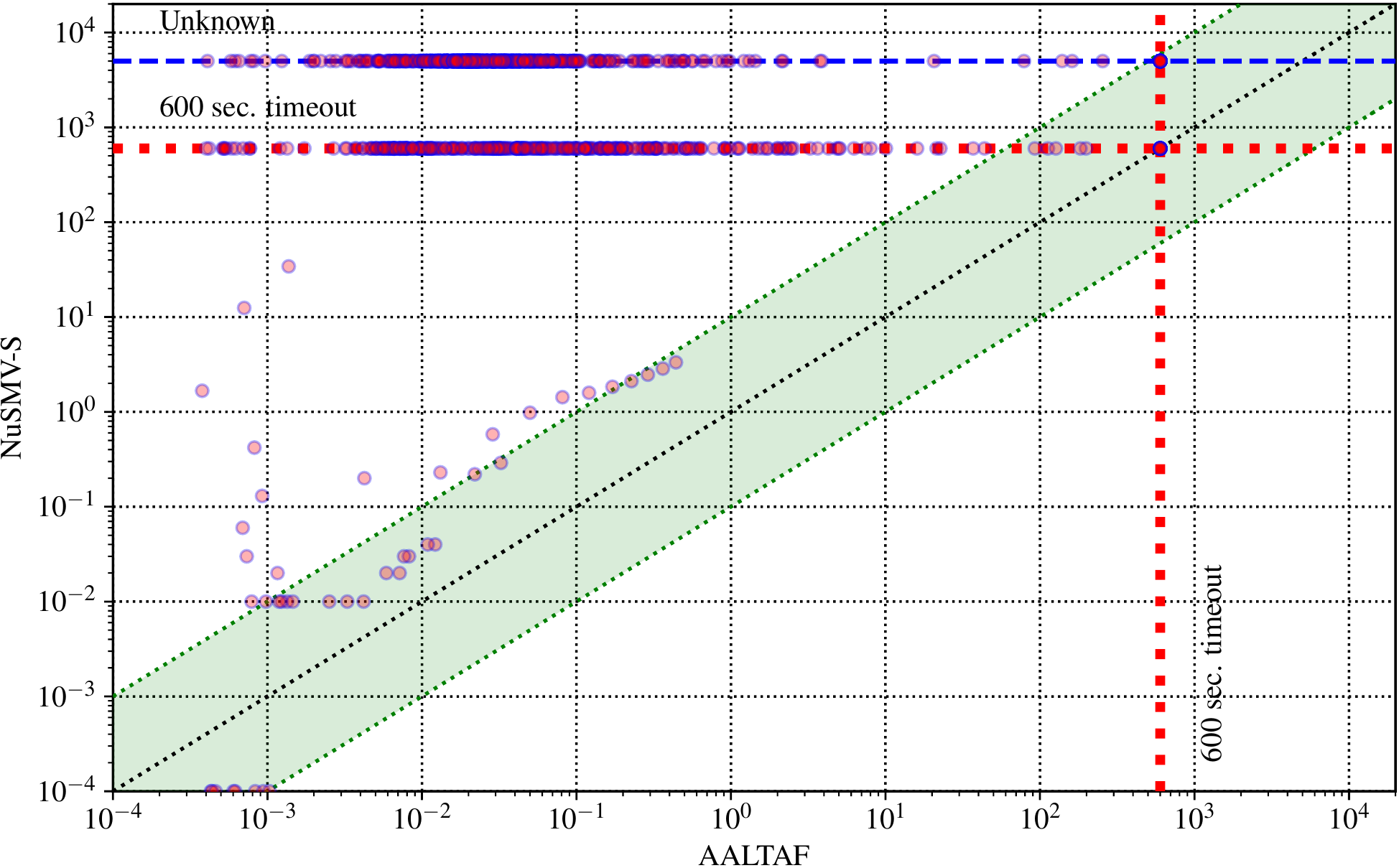} \\
    (a) & (b) \\
    \includegraphics[width=0.48\textwidth]{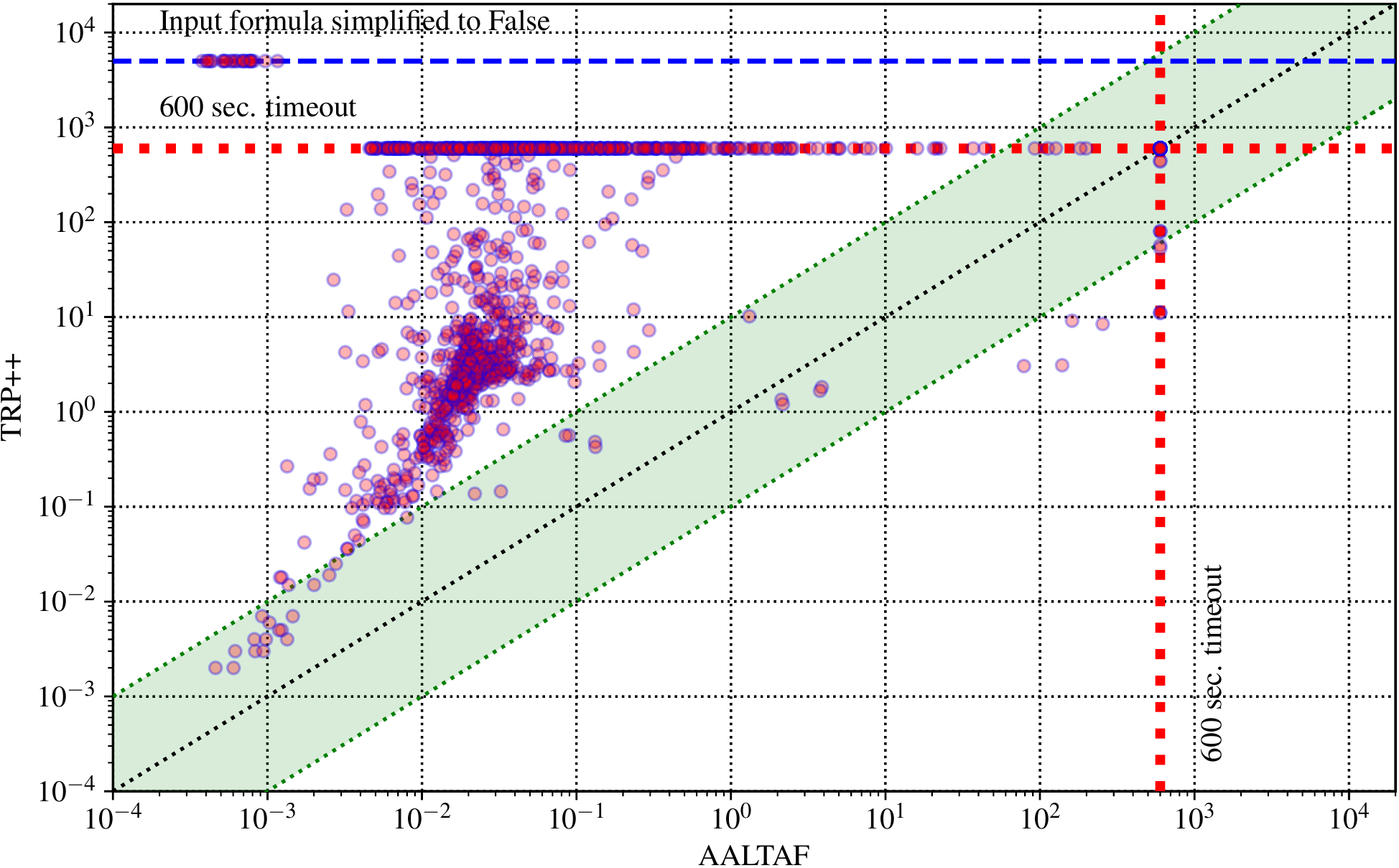} &~
    \includegraphics[width=0.48\textwidth]{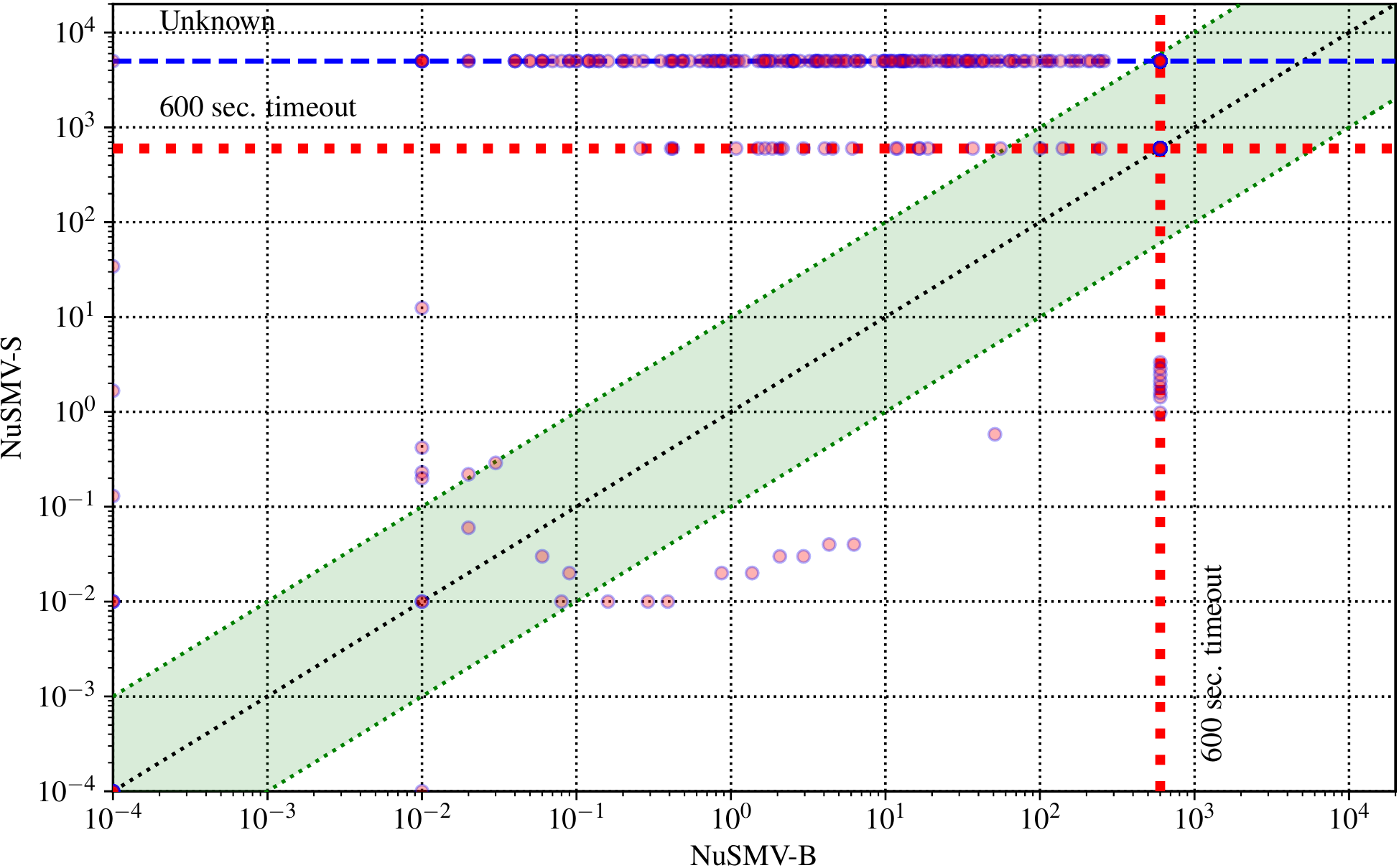} \\
    (c) & (d) \\
    \includegraphics[width=0.48\textwidth]{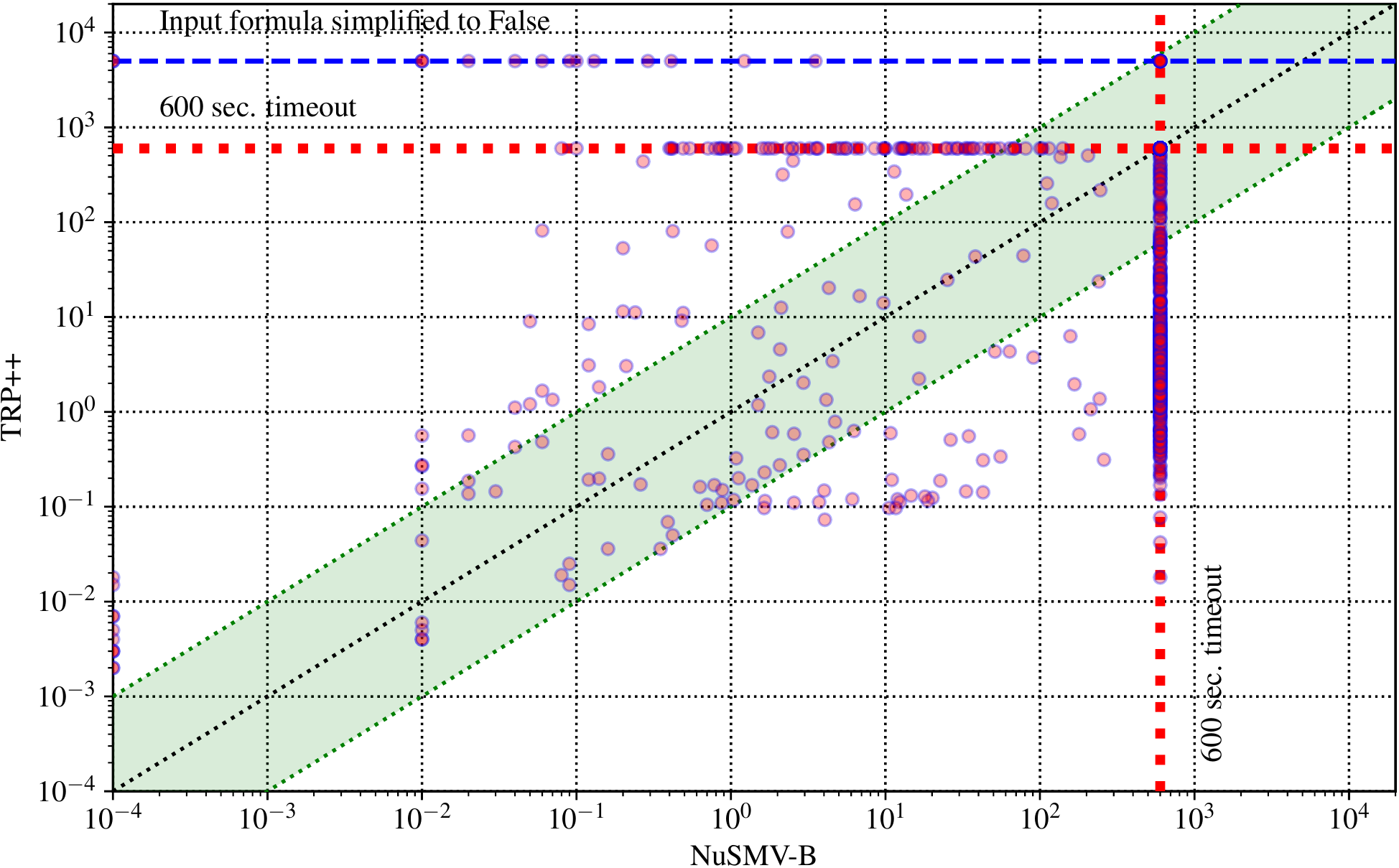}  &~
    \includegraphics[width=0.48\textwidth]{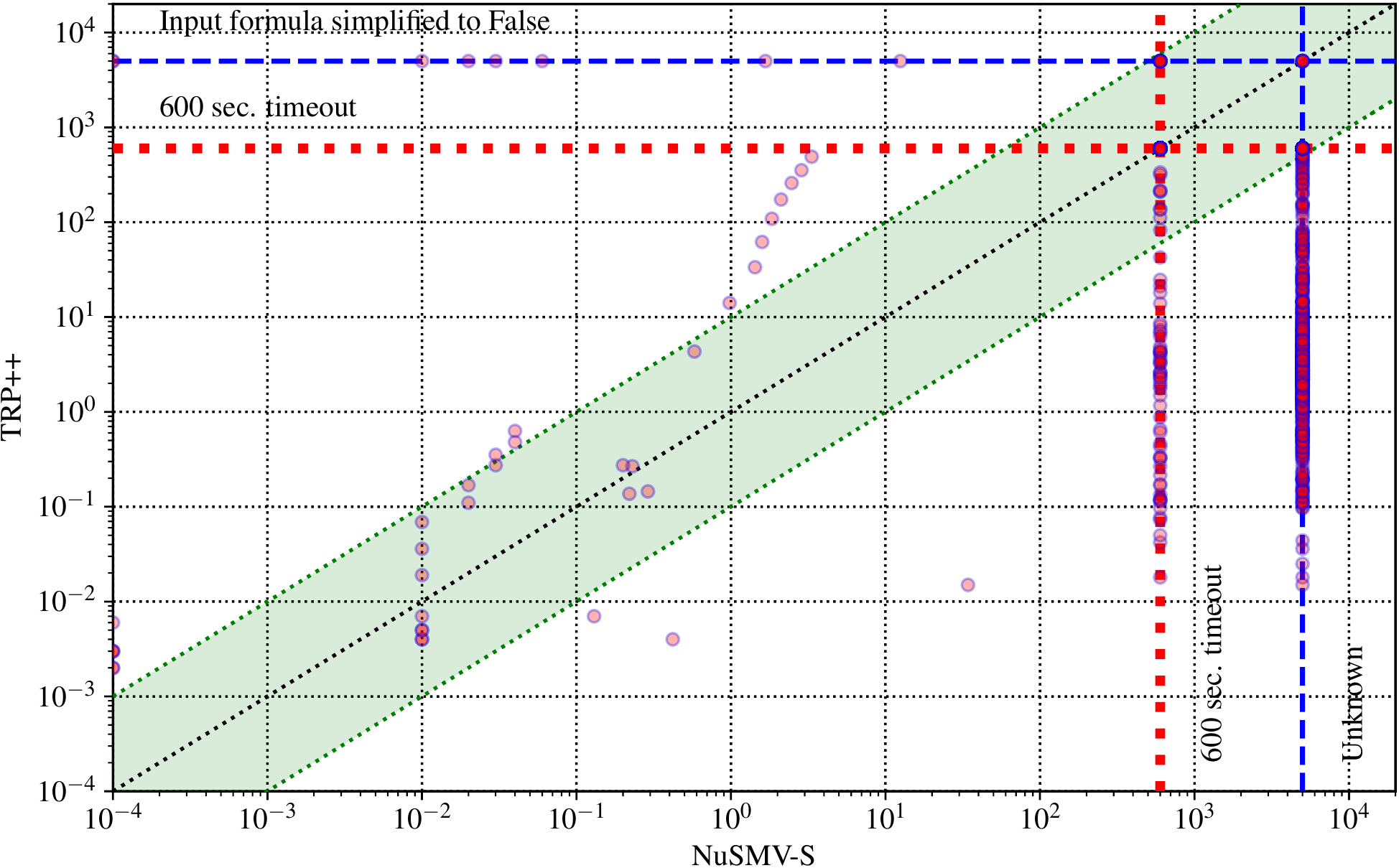}  \\
    (e) & (f)\\
  \end{tabular}
  \caption{\mr{Scatter plots comparing search timings for each algorithm pair.}}
  \label{fig:scatter-time}
\end{figure}

\mr{
\cdc{Figure~\ref{fig:scatter-time} illustrates pairwise comparisons of time efficiency of the \mr{considered} tools}. In particular, Fig.~\ref{fig:scatter-time}(a)
compares \aaltaf with \nusmv-B, Fig.~\ref{fig:scatter-time}(b)
compares \aaltaf with \nusmv-S, Fig.~\ref{fig:scatter-time}(c)
compares \aaltaf with \trppp,
Fig.~\ref{fig:scatter-time}(d) compares \nusmv-B with \nusmv-S,
Fig.~\ref{fig:scatter-time}(e) compares \nusmv-B with \trppp,
and finally Fig.~\ref{fig:scatter-time}(f) compares \nusmv-S with \trppp.
\cdf{Figure~\ref{fig:scatter-time}(c), for instance, shows that \aaltaf outperforms \trppp: most of the points, indeed, are located \cdc{above} the diagonal, thus indicating that \aaltaf requires less time than \trppp to return the unsatisfiable cores. The plot also shows that \trppp exceeds the timeout in several cases (points on the red line \cdc{marked with ``600 sec.\ timeout''}). \cdc{Furthermore, we observe that \trppp operates a pre-processing phase on the input specification prior to the actual identification of \UCs. If it manages to reduce the given set of conjuncts to false at that stage, it stops the computation and raises an alert. The points lying on the line marked with ``Input formula simplified to False'' indicate those cases. Notice that, this simplification occurred in \num{26} cases overall, as depicted in Figs.~\ref{fig:scatter-time}(c),~(e)~and~(f).}} 
\cdf{Moreover,} \nusmv-B, \trppp, \nusmv-S, and \aaltaf
reached the timeout \cdf{overall} in \num{1158}, \num{624}, \num{358} and \num{85}
cases, respectively.
\nusmv-S was able to conclude that the formula was unsatisfiable only in
\num{53} cases, and returned an unknown answer \mr{(i.e., it reached $k=50$ without being
able to decide on unsatisfiability)} 
in \num{985} cases.
%
\cdf{We can \cdc{conclude} that, in terms of \cdc{computation time, \aaltaf outperforms the other three tools}, \nusmv-B \cdc{requires less computation time than} \nusmv-S and \trppp \cdc{in the majority of cases}, and \trppp reveals \cdc{slightly faster and less subject to timeouts} 
than \nusmv-S.}

%
}

\begin{figure}[t!]
	\centering
	\begin{tabular}{@{}cc@{}}
		\begin{tabular}{c}\includegraphics[width=0.35\textwidth]{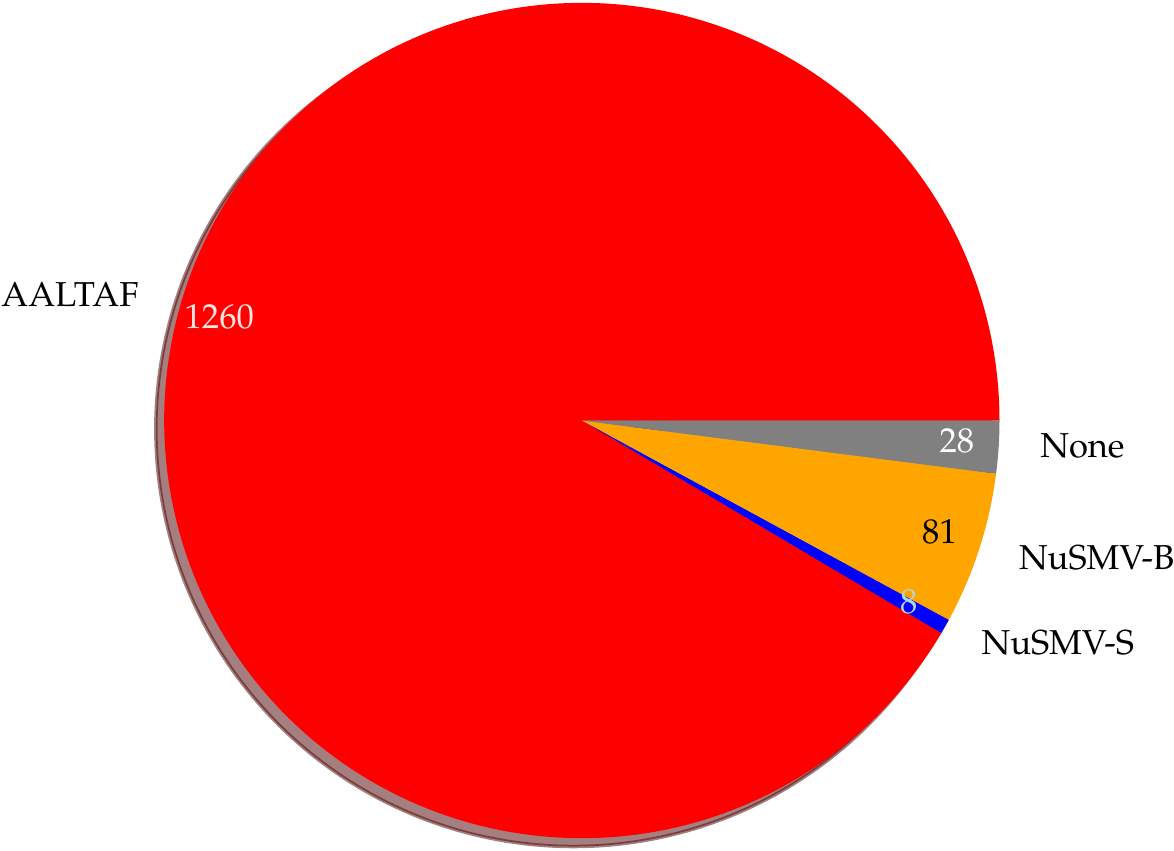}\end{tabular} &  
		\begin{tabular}{c}\includegraphics[width=0.55\textwidth]{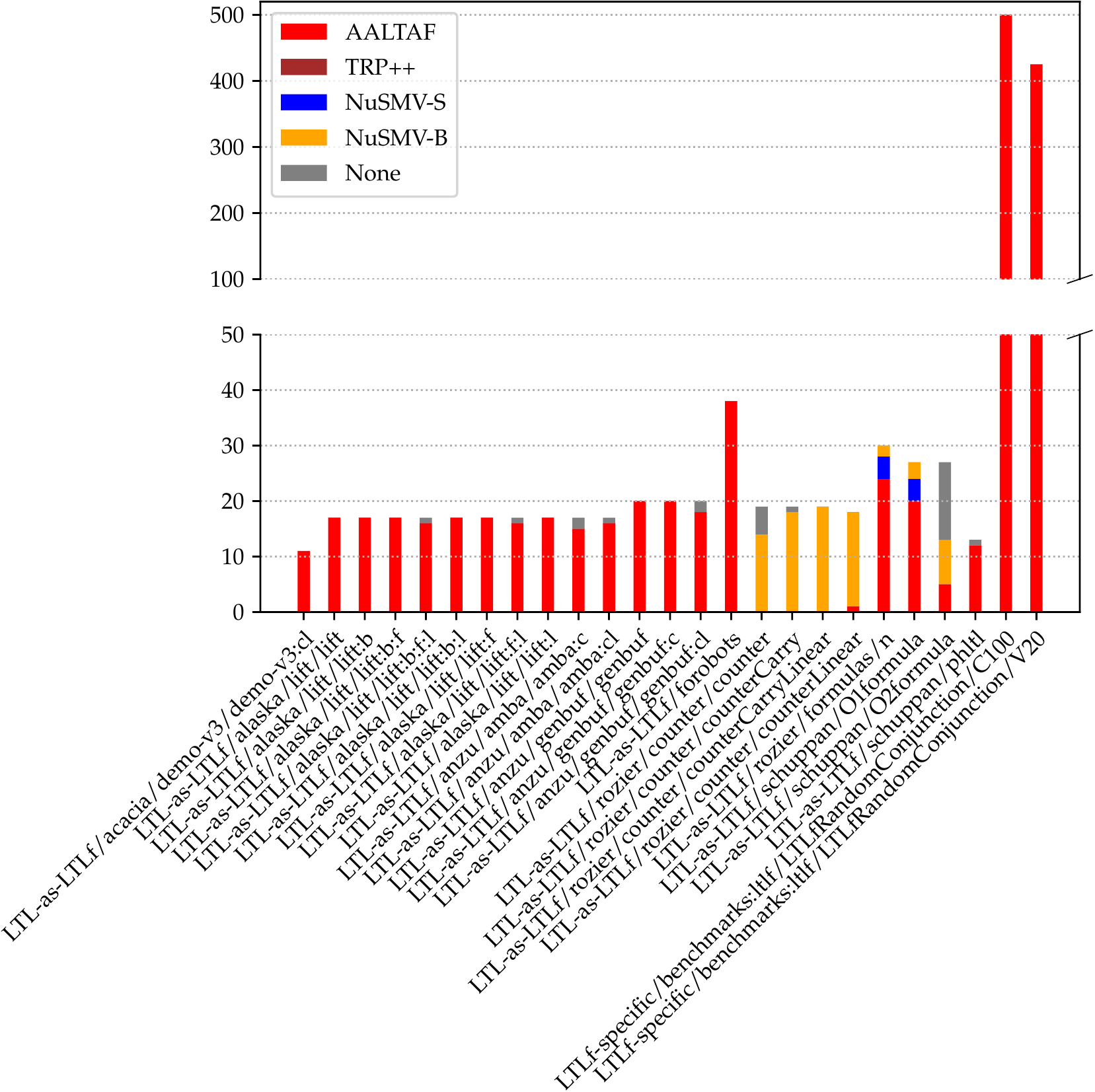}\end{tabular} \\  
		(a) & (b)
	\end{tabular}
	\caption{Pie chart reporting how the different solvers contributed
		to the virtual best (a); Contribution of each solver to the
		virtual best per benchmark family (b), i.e. for each family we
		report the number of times each algorithm was performing best.}
	\label{fig:time-pie-charts}
\end{figure}
%
%
%
%
%
\cdc{Henceforth, we consider the sole cases in which the tools were able to return an $\UC$ within the given resource limits -- thus excluding timeouts, unknown answers and simplifications to false. 
Figure~\ref{fig:time-pie-charts} focuses on computation time:  
\begin{enumerate*}[label=(\roman*)]
	\item the pie chart in Fig.~\ref{fig:time-pie-charts}(a) gives an overview of the number of tests in which a tool was the fastest, and
	\item the stacked bar chart in Fig.~\ref{fig:time-pie-charts}(b) illustrates the results grouped by benchmark family.
\end{enumerate*}

The pie chart in Fig.~\ref{fig:time-pie-charts}(b) confirms that \aaltaf is the most time-efficient tool 
as evidenced by its \num{1260} fastest runs, followed by \nusmv-B (\num{81}), and \nusmv-S (\num{8}). In \num{28} cases, no tool was able to return an unsatisfiable core.
As shown by Fig.~\ref{fig:time-pie-charts}(b), \cdf{\nusmv-B}
 turned out to find the $\UC$ in minimum time with the 
\textsl{LTL-as-LTLf/rozier/counter/*},
benchmark families.
}
\cdf{Moreover, the  plot also shows that the problems of the \textsl{LTL-as-LTLf/schuppan/O2Formula} benchmark family are the most challenging ones for \cdc{all the implemented techniques}. Indeed, \num{14} out of the \num{28} problems that were not solved by any tools belong to this benchmark family.}
\mr{%
\cdf{Moreover,} the analysis
\cdf{per} benchmark family (see
Fig.~\ref{fig:time-pie-charts}(b)) confirms the superiority \cdc{in terms of time efficiency of}
\aaltaf in many benchmark families.
%
%
%

\begin{figure}[tbp!]
	\centering
	\begin{tabular}{@{}c@{}c@{}}
		\includegraphics[width=0.45\textwidth]{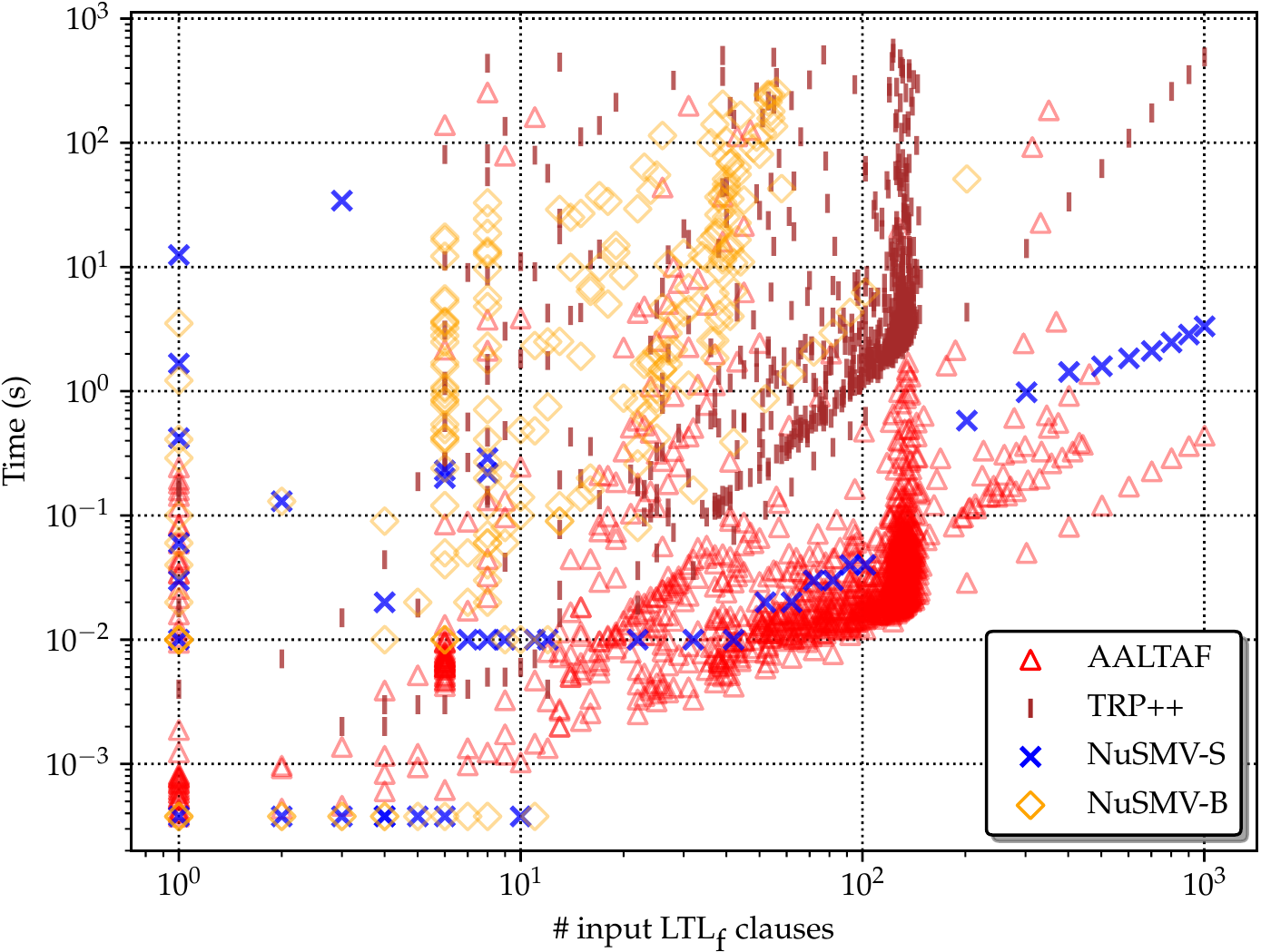}
	\end{tabular}
	\caption{\mr{Number of conjuncts in $\Gamma$ and respective computation time (in seconds) for each algorithm.}}
	\label{fig:clauses-vs-time}
\end{figure}
\begin{figure}[tbp!]
	\centering
	\begin{tabular}{@{}c@{}c@{}c@{}}
		\includegraphics[width=0.33\textwidth]{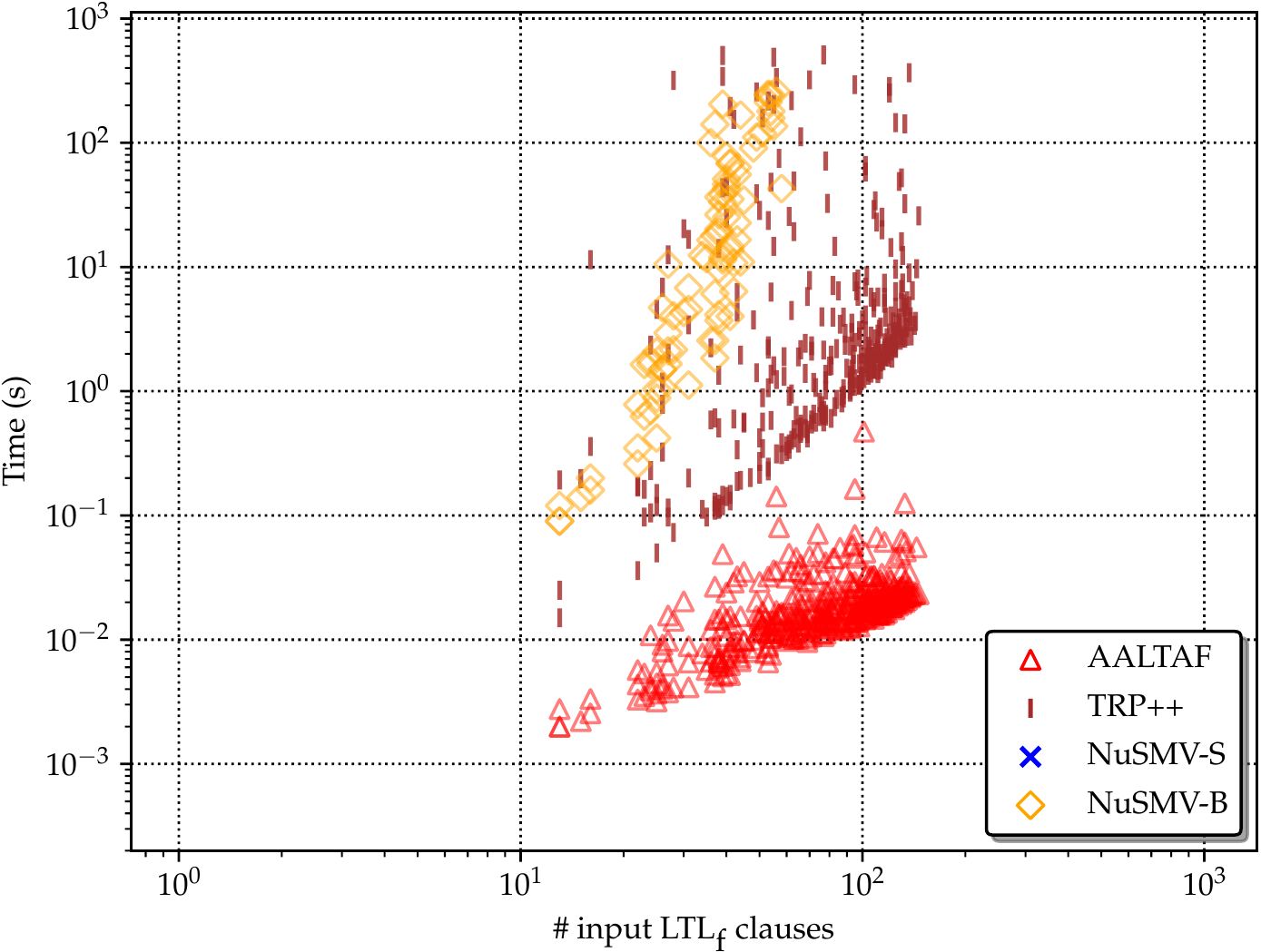} &~
		\includegraphics[width=0.33\textwidth]{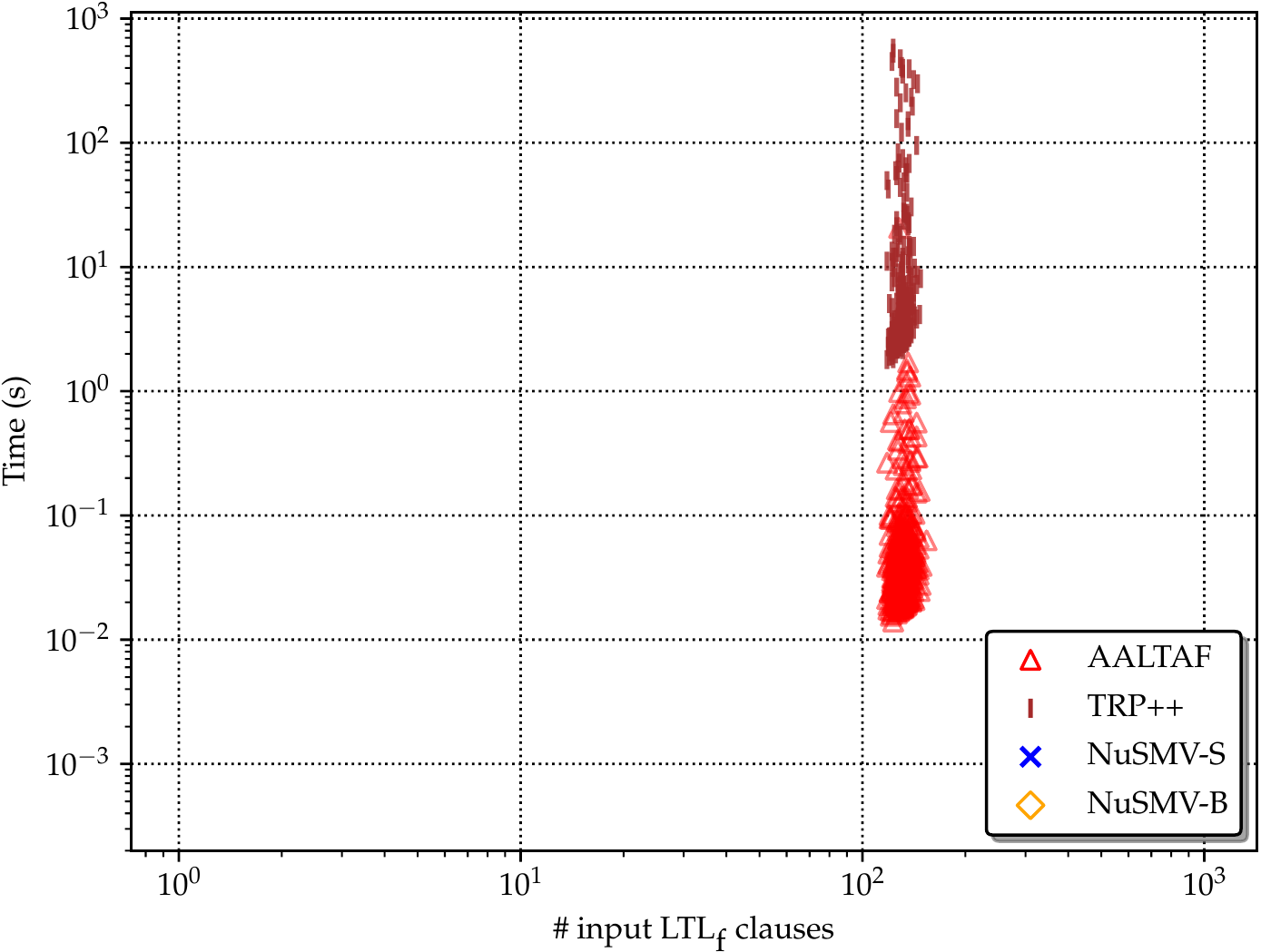} &~
		\includegraphics[width=0.33\textwidth]{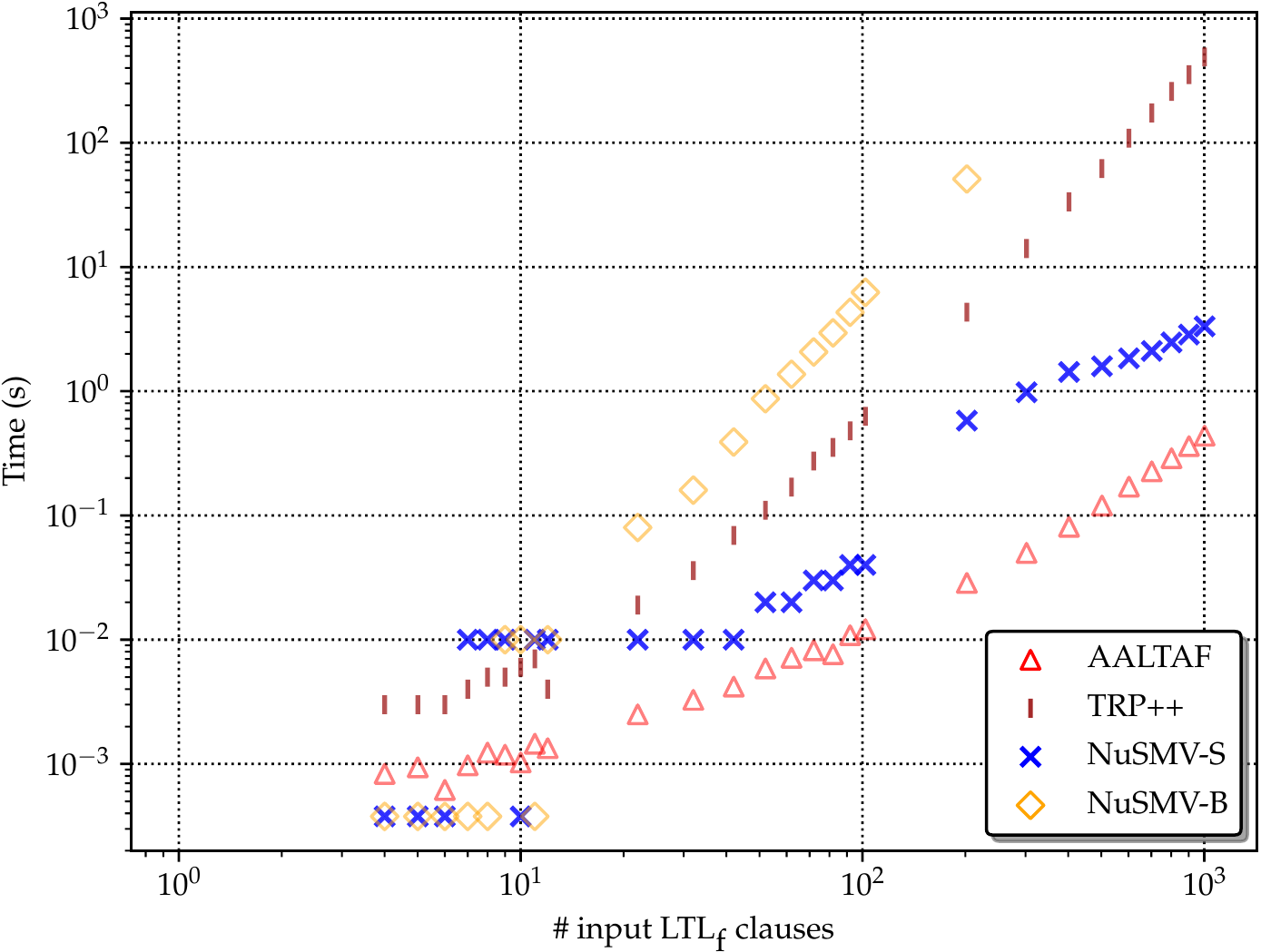} \\
		(a) & (b) & (c)\\
	\end{tabular}
	\caption{\mr{Number of conjuncts in $\Gamma$ and respective computation time (in seconds) for each algorithm for three categories:
			(a) \textsl{LTLf-specific/benchmarks/LTLFRandomConjunction/V20},
			(b) \textsl{LTLf-specific/benchmarks/LTLFRandomConjunction/C100}, and
			(c) \textsl{LTL-as-LTLf/schuppan/O1Formula}.}}
	\label{fig:clauses-vs-time-cat}
\end{figure}
\cdf{In order to further inspect the correlation between the time performance of the tools and the type of problems solved, we analyzed the relationship between the number of conjuncts of the problems and the corresponding computation time.}
Figure~\ref{fig:clauses-vs-time} plots the number of conjuncts in
$\Gamma$ (i.e., its cardinality) against the respective computation time
(in seconds) for each of the considered algorithms. 
Figure~\ref{fig:clauses-vs-time-cat} 
\cdc{isolates the points} stemming from \cdf{three} families \cdc{in particular}:
  \textsl{LTLf-specific/benchmarks/LTLFRandomConjunction/V20}
  (Fig.~\ref{fig:clauses-vs-time-cat}(a)),
  \textsl{LTLf-specific/benchmarks/LTLFRandomConjunction/C100}
  (Fig.~\ref{fig:clauses-vs-time-cat}(b)), and
  \textsl{LTL-as-LTLf/schuppan/O1Formula}
  (Fig.~\ref{fig:clauses-vs-time-cat}(c)).
\cdf{The plots show that a relationship exists between the number of \LTLf clauses and the computation time for all the four tools: the required time overall increases when the number of clauses increases. However, the number of clauses is not the only factor affecting the computation time. For instance, for the \textsl{LTLf-specific/benchmarks/LTLFRandomConjunction/C100} benchmark family (Fig.~\ref{fig:clauses-vs-time-cat}(b)), the computation time varies independently of the number of conjuncts,} \cdc{which ranges in a short interval (\num{118} to \num{154} clauses, as per Table~\ref{tab:families}). Also, we can observe that\cdf{, with this benchmark family,} neither \nusmv-S nor \nusmv-B could return an $\UC$ under the imposed experimental resource constraints, while \aaltaf appears to be faster than \trppp, following the general trend.
Figures~\ref{fig:clauses-vs-time-cat}(a)~and~\ref{fig:clauses-vs-time-cat}(c) illustrate the different rapidity with which curves increase with the number of conjuncts: the sharpest one is associated to \nusmv-B, followed by \trppp and \nusmv-S (the latter performing better than the former with smaller sets of conjuncts, though, as depicted in Fig.~\ref{fig:clauses-vs-time-cat}(c)).
The most gradual upward trend belongs to the curve of \aaltaf.
} 
\end{marcodo}

\begin{figure}[tbp]
	\centering
	\begin{tabular}{@{}cc@{}}
		\begin{tabular}{c}\includegraphics[width=0.35\textwidth]{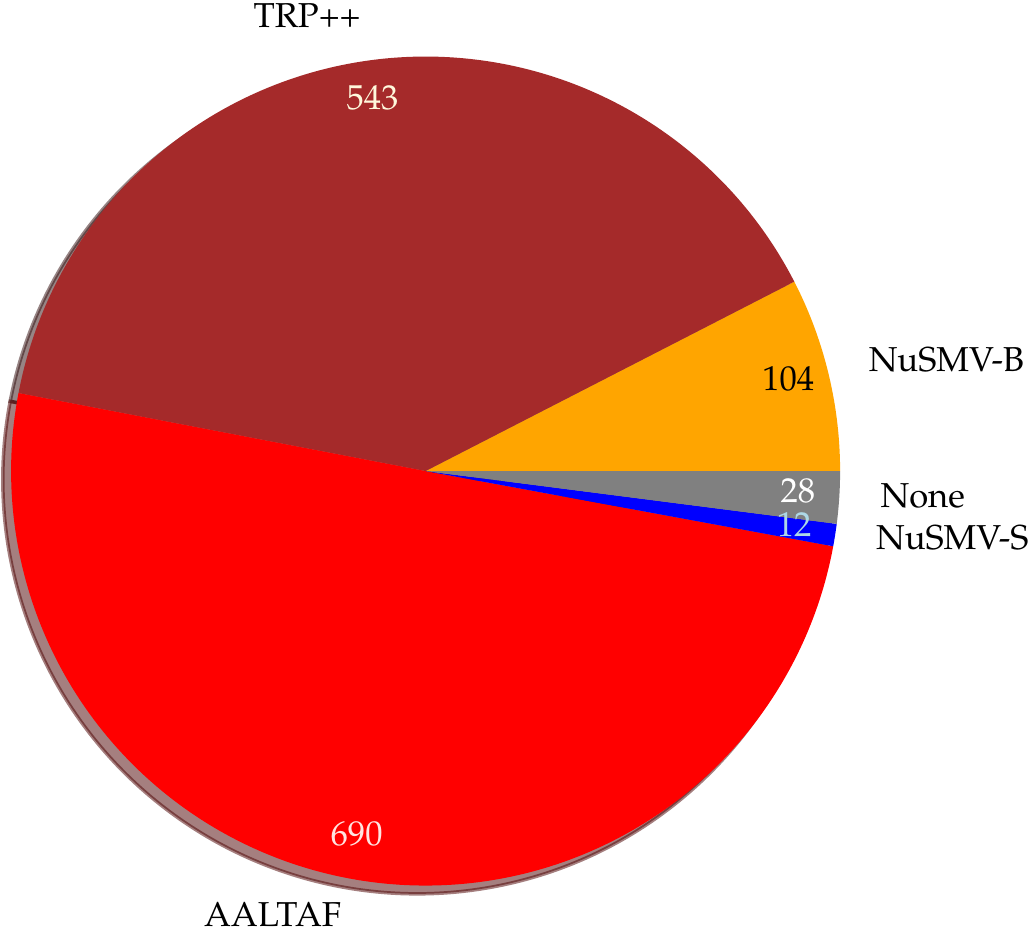} \end{tabular} &  
		\begin{tabular}{c}\includegraphics[width=0.55\textwidth]{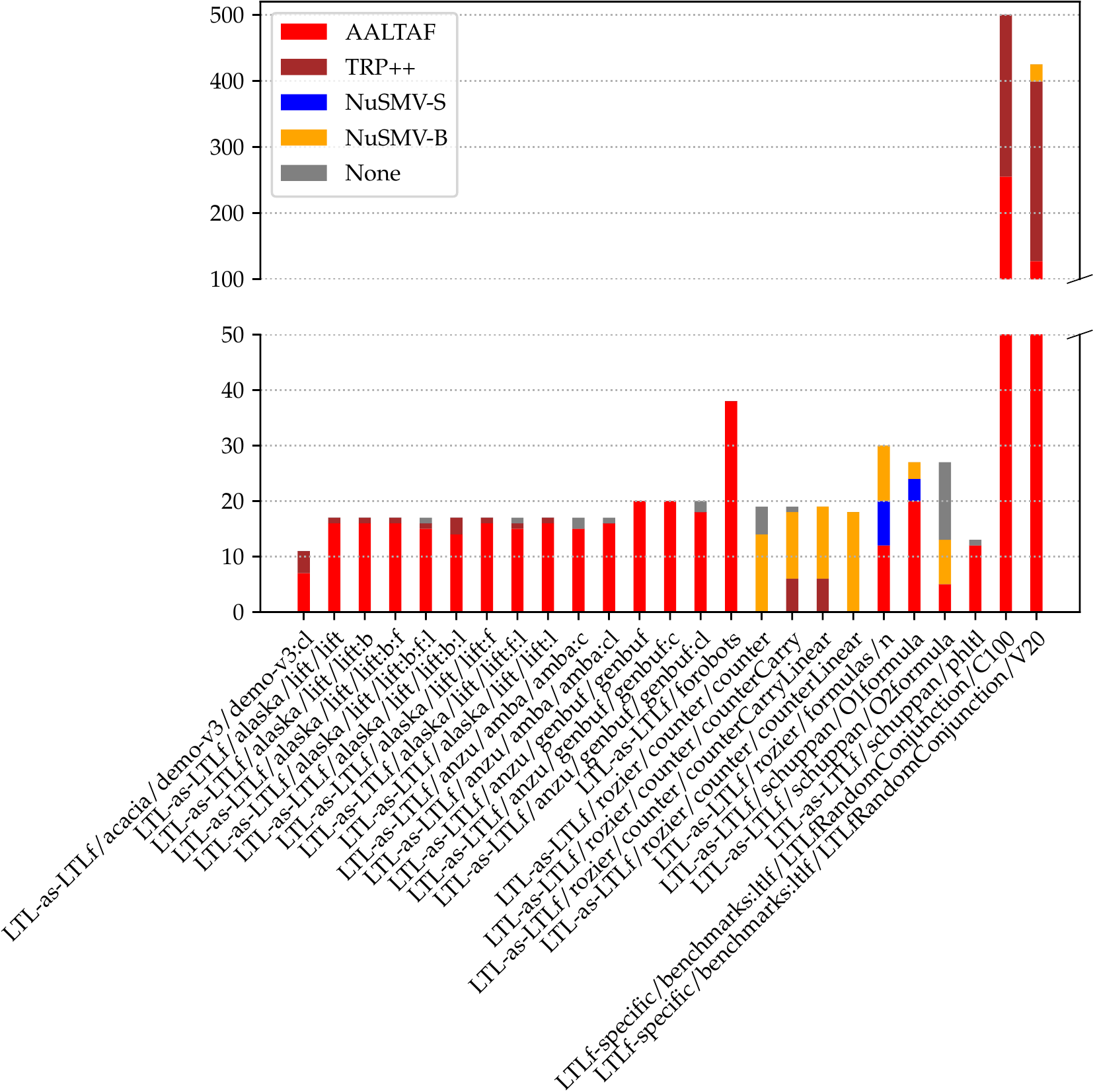} \end{tabular} \\  
		(a) & (b) \\
	\end{tabular}
	\caption{Number of benchmarks in which the solver returned an
		unsatisfiable core of the smallest size (a) for the entire set of benchmarks,
		and (b) per benchmarks family.}
	\label{fig:size-pie-charts}
\end{figure}
%
%
%
%
As far as the cardinality of the extracted \UCs is concerned,
\mr{Fig.~\ref{fig:size-pie-charts} depicts the result of our analysis in the
different benchmarks families.}
As shown
in the pie chart of Fig.~\ref{fig:size-pie-charts}, we report that
\aaltaf, \trppp, \nusmv-B, and \nusmv-S extract \UCs
that are the smallest in size%
\footnote{\mr{Notice that, by ``smallest'' we mean the unsatisfiable core of smallest
		cardinality among the ones computed by the solvers. Notice that, the smallest \UC does not
		necessarily correspond to the minimum one, as discussed when
		presenting the algorithms.}
	\cdc{Also, we remark that when two solvers return an \UC of the
	same size, we associate the best result to the tool that took the lowest time
	to return it.}} in \cdc{\num{690}, \num{543}, \num{104}}
and
\num{12} cases, respectively.
\cdc{The overall minimum, maximum, average, and median cardinality
of the smallest computed \UCs were \num{1}, \num{74}, \num{6.615}, and \num{4}, respectively.}

\cdc{We observe that}
\mr{%
	for the \textsl{LTL-as-LTLf/rozier/counter/*} benchmarks, \cdf{\nusmv-B}
	computes the unsatisfiable core with the smallest size \cdc{in the majority of cases}.
	Notice that, it is also
	the one that \cdc{most often} performs best in terms of search time (see
	Fig.~\ref{fig:time-pie-charts}(b)).
	%
}
\cdf{\nusmv-B is also able to obtain the smallest \UC with} \cdc{most of the benchmarks within the \textsl{LTL-as-LTLf/schuppan/O2formula} family.}
\cdf{On all other benchmarks, \aaltaf outperforms the other algorithms, and \cdc{with the \textsl{LTLf-specific/*} benchmarks, \trppp is the second best solver to find the smallest \UCs after \aaltaf.}
}
%
These results suggest that \nusmv-B could be preferred on benchmarks
with fewer propositional variables and larger temporal depth.
However, the SAT based approaches seem to work
better on benchmarks with a higher number of propositional variables that are not
always directly correlated with one another. Indeed, in these last
cases, BDDs may suffer a blow-up in size due to the canonicity of the
representation (BDD dynamic variable ordering could help 
\cdc{though to a limited extent}~\cite{DBLP:conf/eurodac/FeltYBS93}).}
\cdf{Notice that, all solvers 
	were not capable of dealing with most of the the big conjunctions of formulas in \textsl{LTL-as-LTLf/schuppan/O2Formula}%
	}
\cdc{(the corresponding tallest stacked bar in Fig.~\ref{fig:size-pie-charts} is labeled as ``None'', indeed).}

%

\begin{figure}[t!]
  \centering
  \begin{tabular}{@{}c@{}c@{}c@{}}
    \includegraphics[width=0.32\textwidth]{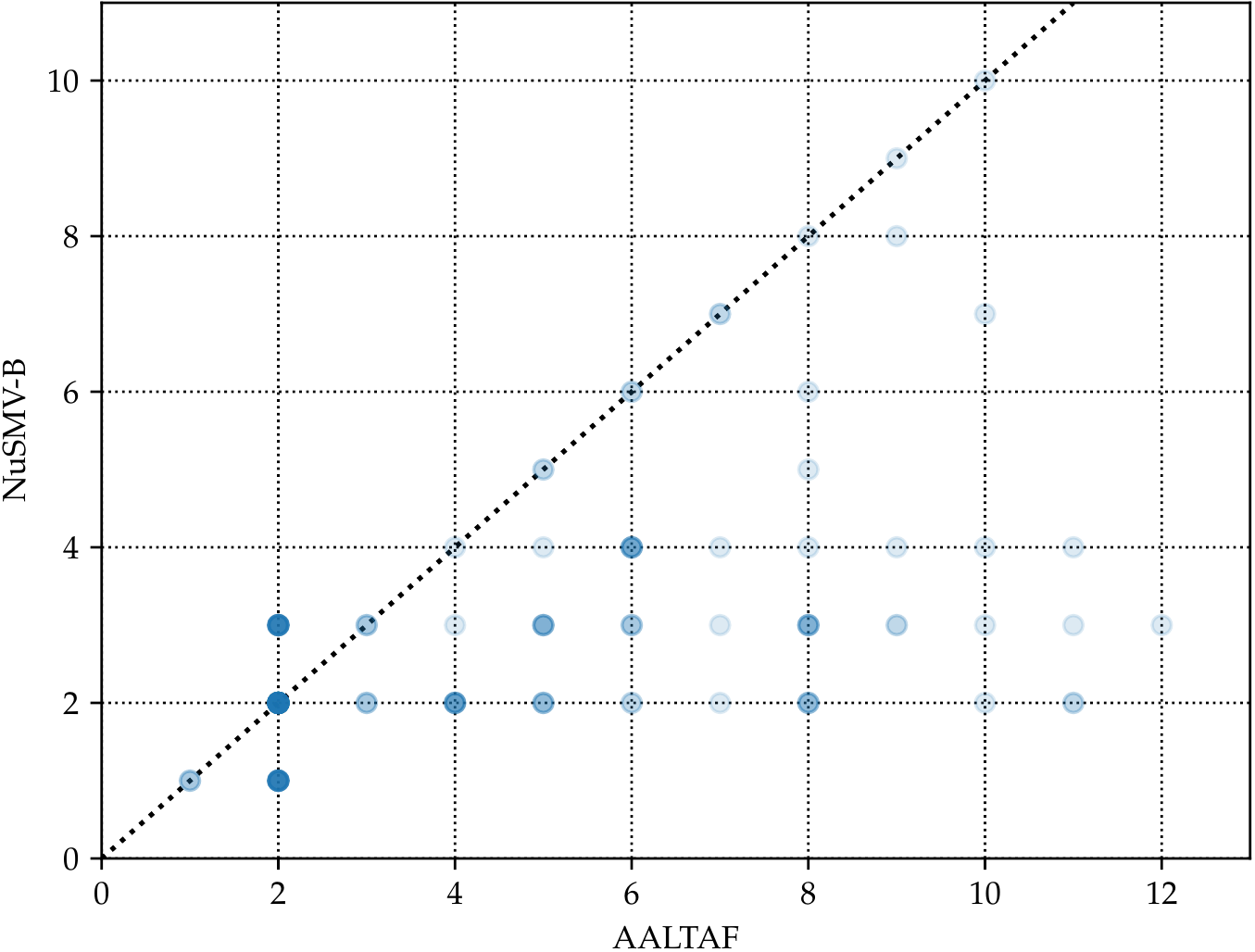} &~
    \includegraphics[width=0.32\textwidth]{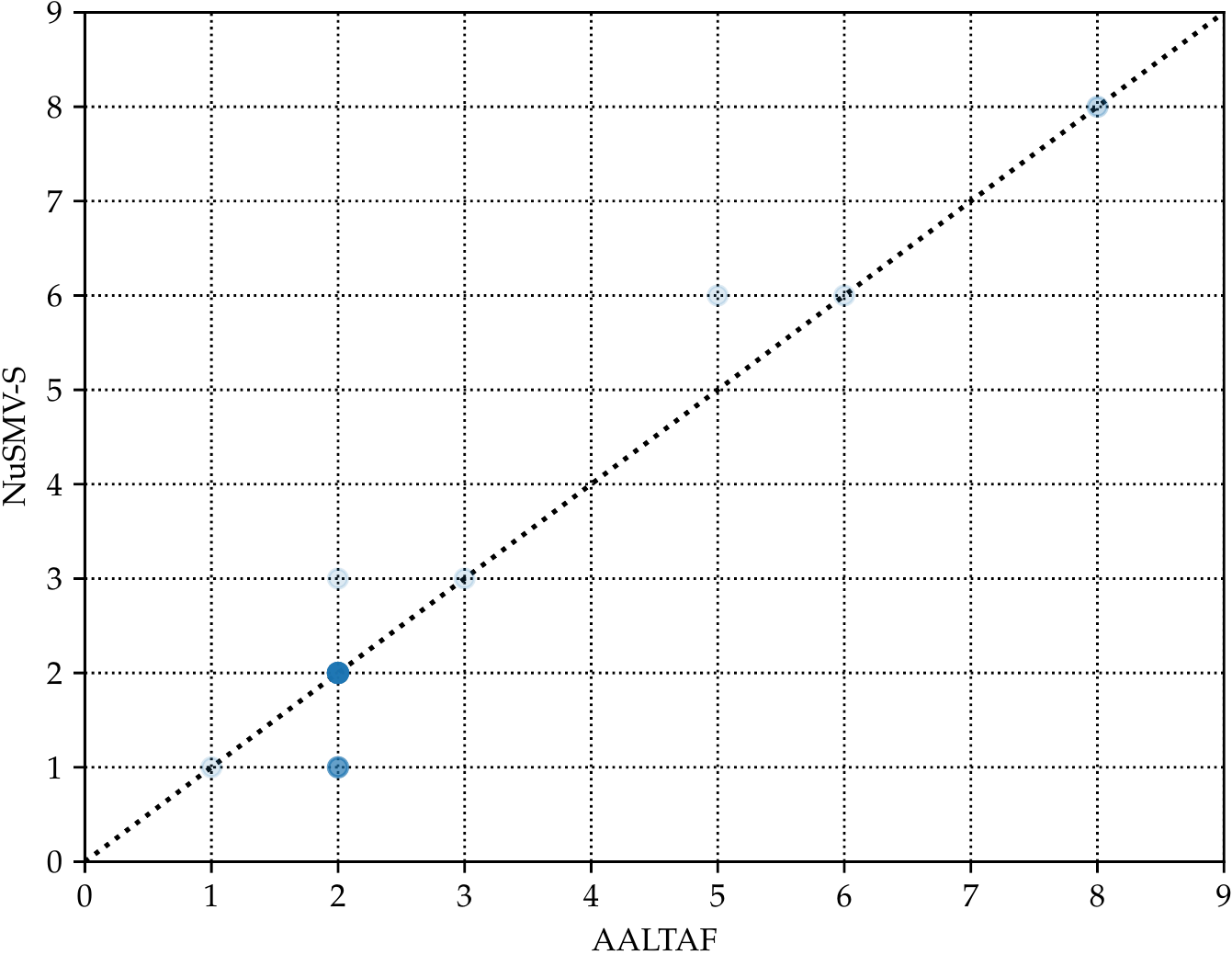} &~
    \includegraphics[width=0.32\textwidth]{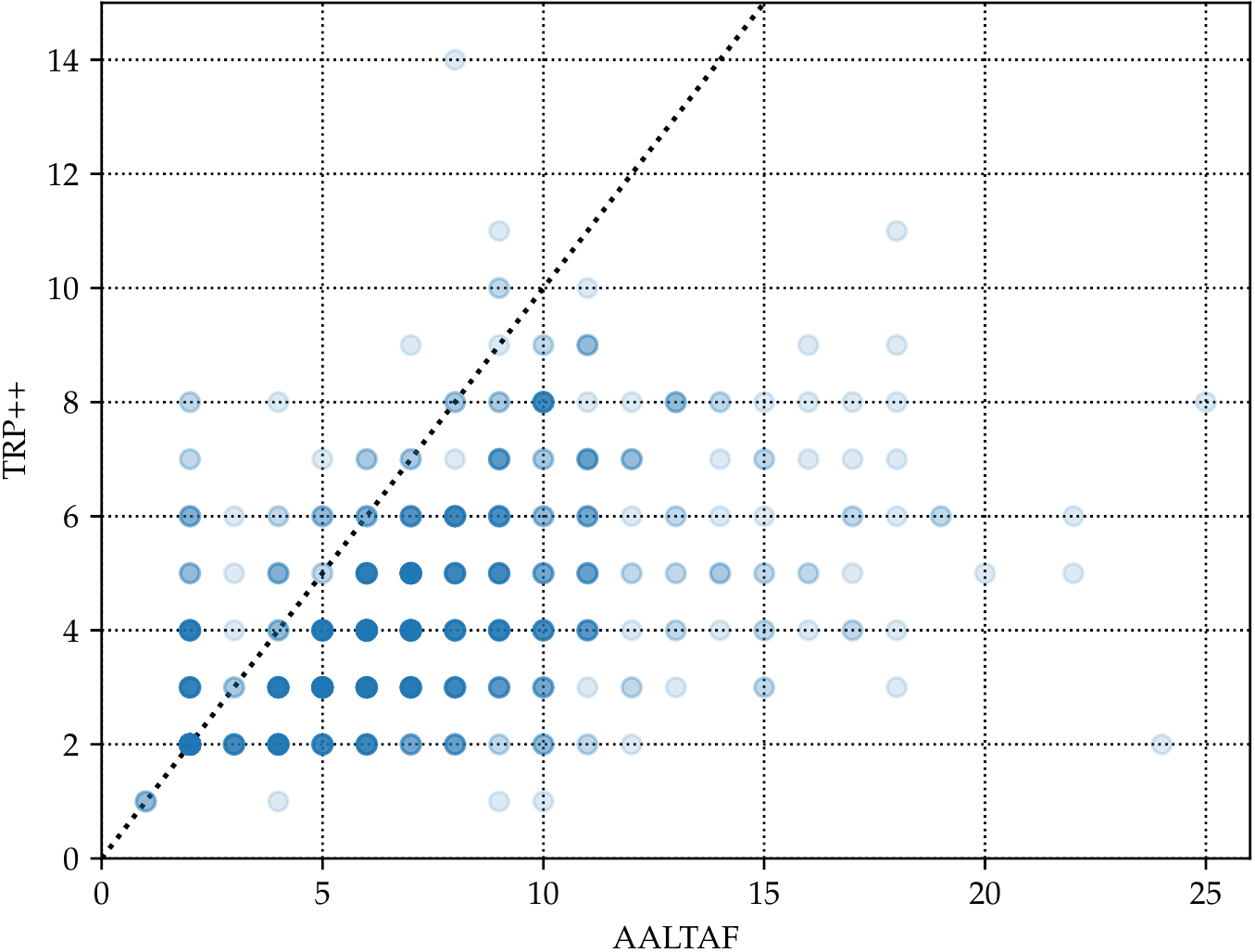} 
     \\
    (a) & (b) & (c)\\
    \includegraphics[width=0.32\textwidth]{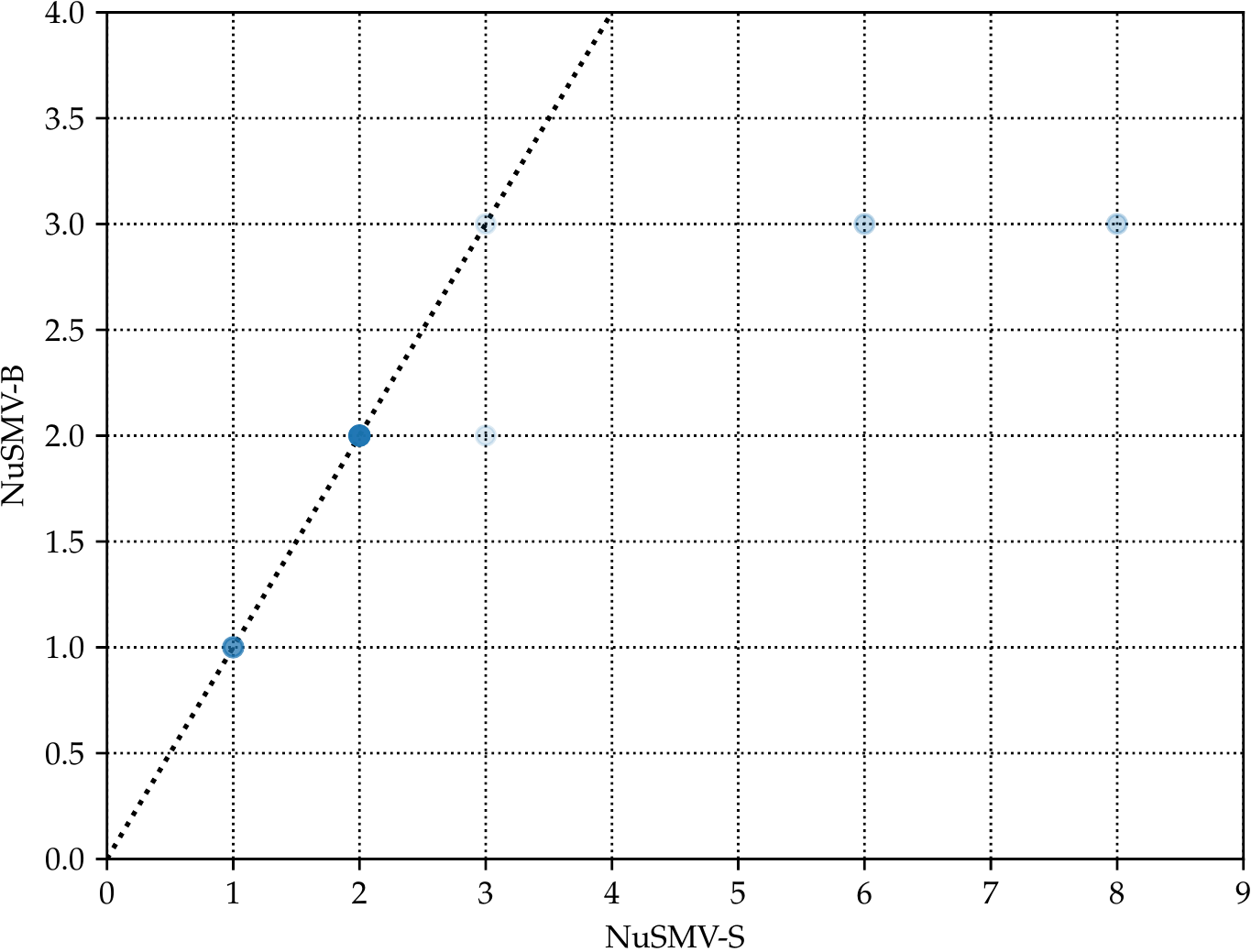} &~
    \includegraphics[width=0.32\textwidth]{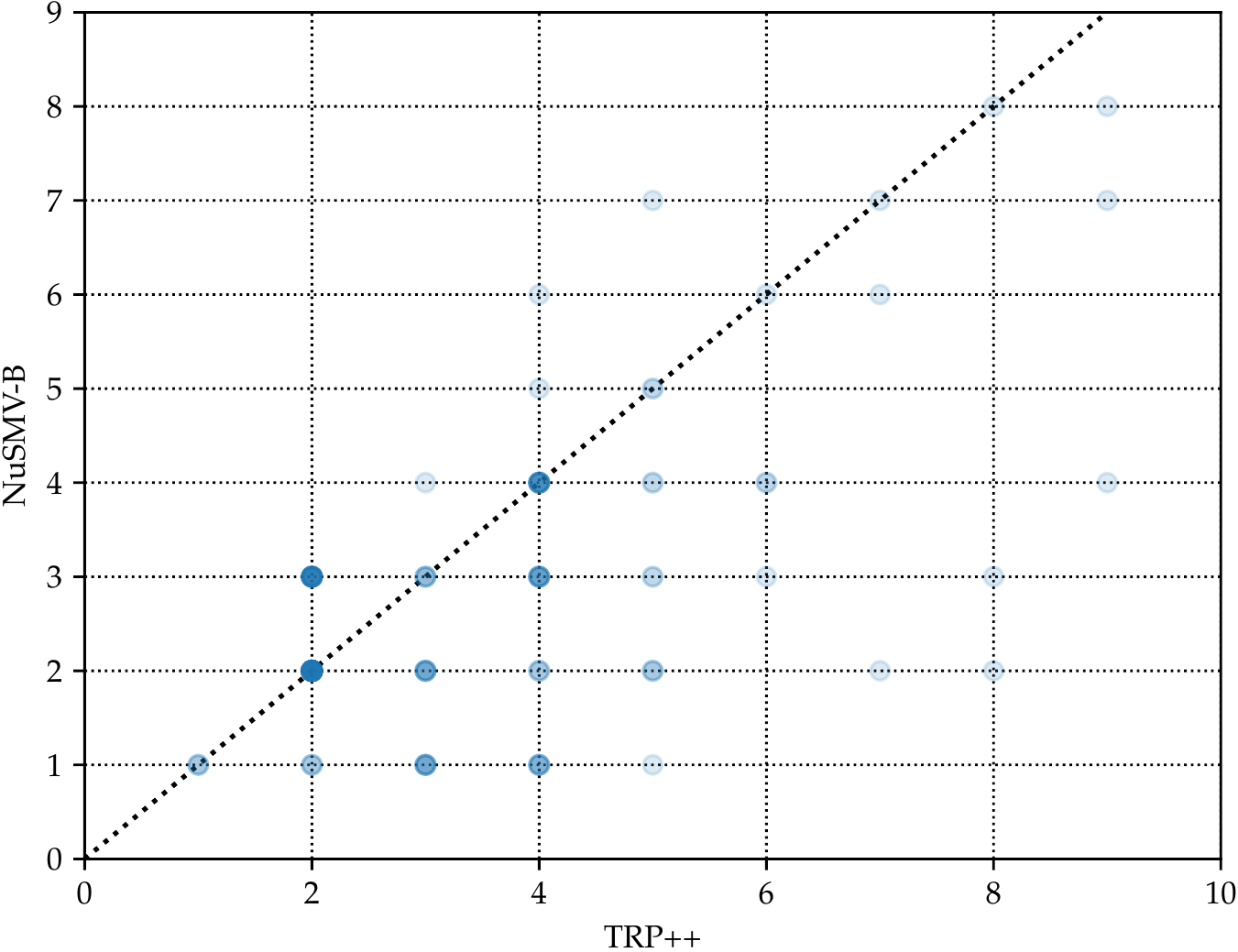}  &~
    \includegraphics[width=0.32\textwidth]{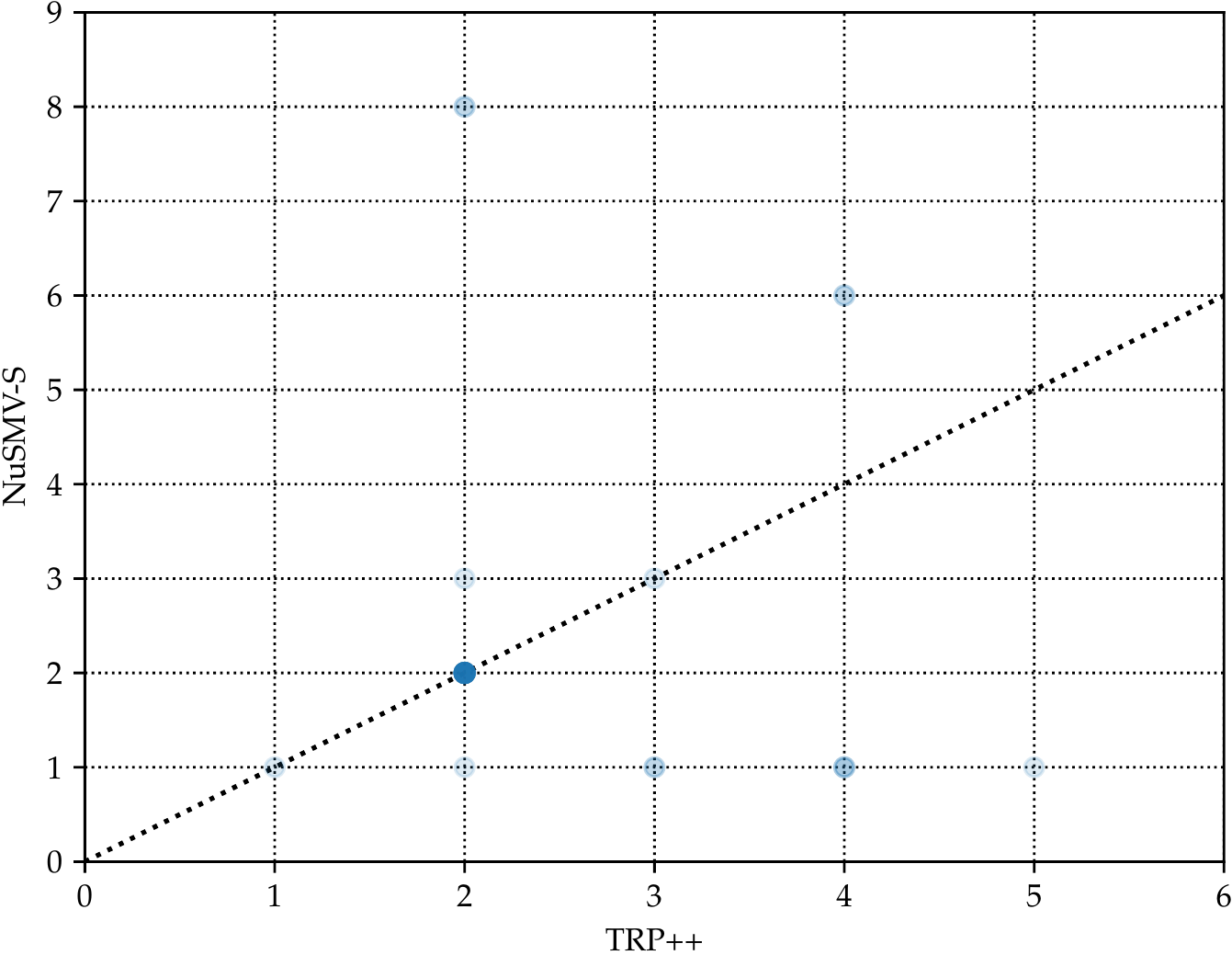}  \\
    (d) & (e) & (f)\\
  \end{tabular}
  \caption{\mr{Scatter plots comparing the cardinality of computed \UCs for each algorithm pair.}}
  \label{fig:scatter-cardinality}
\end{figure}


Figure~\ref{fig:scatter-cardinality}
plots the
\mr{pairwise comparison between different tools} on the subset of the
cases where both approaches were able to compute the \UC.
\cdf{For instance, Fig.\ref{fig:scatter-cardinality}(a) compares the cardinality of the \UCs returned by \aaltaf and the cardinality of the \UCs returned by \nusmv-B. The plot shows that \nusmv-B \UCs are smaller than the ones returned by \aaltaf: most of the points are indeed located below the diagonal. The intensity of the points represents the number of cases for which the two algorithms returned \UCs with the specific cardinalities corresponding to the coordinates of the point in the plot. Overall, we can observe that when a solution is returned, the cardinality of the \UCs returned by \trppp and \nusmv-B is often lower than the cardinality of the \UCs returned by \aaltaf.}
%


\mr{To conclude, we remark that these results
\begin{enumerate*}[label=(\roman*)]
\item evidence an overall better performance of \aaltaf both in terms of time efficiency and cardinality of the extracted \UCs, and
\item emphasise a complementarity of the proposed approaches.
\end{enumerate*}
\cdc{Table~\ref{tab:of:fame} summarizes the findings above}. 
Observe that \cdf{none of the algorithms outperforms all the others on every benchmark. For example, }
 \nusmv-S and \nusmv-B end up in a timeout and return an unknown answer in considerably many cases,
\cdf{so that} a number of problems are solved by only one of them.
\cdc{\aaltaf does not always turn out to return the smallest \UC: 
in a number of cases, \trppp, \nusmv-B and \nusmv-S extract \UCs of a lower cardinality,
excelling in particular in those cases in which \aaltaf ends in a timeout.
A deeper investigation of the characteristics that lead to such behaviors paves the path for future work.}
}
%

\begin{table}[tbp]
  \caption[Tool performance]{Best results as per the cardinality of \UCs and wall-clock timings.}
  \label{tab:of:fame}
  \begingroup
  \renewcommand*{\arraystretch}{1.4}
  \begin{adjustbox}{max width=\textwidth}
\begin{tabular}{@{}l S[table-format=4.0] S[table-format=5.0]@{ (}S[table-format=3.2]@{\,\%)} S[table-format=5.0]@{ (}S[table-format=3.2]@{\,\%)}    S[table-format=5.0]@{ (}S[table-format=3.2]@{\,\%)} S[table-format=5.0]@{ (}S[table-format=3.2]@{\,\%)}    S[table-format=5.0]@{ (}S[table-format=3.2]@{\,\%)} S[table-format=5.0]@{ (}S[table-format=3.2]@{\,\%)}    S[table-format=5.0]@{ (}S[table-format=3.2]@{\,\%)} S[table-format=5.0]@{ (}S[table-format=3.2]@{\,\%)}    S[table-format=5.0]@{ (}S[table-format=3.2]@{\,\%)} }
	\toprule
	                                              &                           &                \multicolumn{4}{c}{\aaltaf}                &                \multicolumn{4}{c}{\trppp}                 &               \multicolumn{4}{c}{\nusmv-S}                &               \multicolumn{4}{c}{\nusmv-B}                &                 \multicolumn{2}{c}{}                  \\ \cmidrule(l){3-6} \cmidrule(l){7-10} \cmidrule(l){11-14} \cmidrule(l){15-18} 
	                                       Family & \multicolumn{1}{r}{Total} & \multicolumn{2}{c}{min.UC} & \multicolumn{2}{c}{min.time} & \multicolumn{2}{c}{min.UC} & \multicolumn{2}{c}{min.time} & \multicolumn{2}{c}{min.UC} & \multicolumn{2}{c}{min.time} & \multicolumn{2}{c}{min.UC} & \multicolumn{2}{c}{min.time} & \multicolumn{2}{c}{None} \\ \cmidrule(r){1-2} \cmidrule(l){3-4} \cmidrule(l){5-6} \cmidrule(l){7-8} \cmidrule(l){9-10} \cmidrule(l){11-12} \cmidrule(l){13-14} \cmidrule(l){15-16} \cmidrule(l){17-18} \cmidrule(l){19-20}
	Overall                                       & 1377                      & 690 & 50.11                & 1260 & 91.50                 & 543 & 39.43                & 0 & 0.00                     & 12 & 0.87                  & 8 & 0.58                     & 104 & 7.55                 & 81 & 5.88                    & 28 & 2.03                                \\ \midrule
	LTL-as-LTLf/acacia/demo-v3/demo-v3:cl         & 11                        & 7   & 63.64                & 11   & 100.00                & 4   & 36.36                & \multicolumn{2}{c}{~}          & \multicolumn{2}{c}{~}       & \multicolumn{2}{c}{~}          & \multicolumn{2}{c}{~}      & \multicolumn{2}{c}{~}         & \multicolumn{2}{c}{~}       \\ \arrayrulecolor{black!30}\midrule
	LTL-as-LTLf/alaska/lift/lift                  & 17                        & 16  & 94.12                & 17   & 100.00                & 1   & 5.88                 & \multicolumn{2}{c}{~}          & \multicolumn{2}{c}{~}       & \multicolumn{2}{c}{~}          & \multicolumn{2}{c}{~}      & \multicolumn{2}{c}{~}         & \multicolumn{2}{c}{~}       \\ \midrule
	LTL-as-LTLf/alaska/lift/lift:b                & 17                        & 16  & 94.12                & 17   & 100.00                & 1   & 5.88                 & \multicolumn{2}{c}{~}          & \multicolumn{2}{c}{~}       & \multicolumn{2}{c}{~}          & \multicolumn{2}{c}{~}      & \multicolumn{2}{c}{~}         & \multicolumn{2}{c}{~}       \\ \midrule
	LTL-as-LTLf/alaska/lift/lift:b:f              & 17                        & 16  & 94.12                & 17   & 100.00                & 1   & 5.88                 & \multicolumn{2}{c}{~}          & \multicolumn{2}{c}{~}       & \multicolumn{2}{c}{~}          & \multicolumn{2}{c}{~}      & \multicolumn{2}{c}{~}         & \multicolumn{2}{c}{~}       \\ \midrule
	LTL-as-LTLf/alaska/lift/lift:b:f:l            & 17                        & 15  & 88.24                & 16   & 94.12                 & 1   & 5.88                 & \multicolumn{2}{c}{~}          & \multicolumn{2}{c}{~}       & \multicolumn{2}{c}{~}          & \multicolumn{2}{c}{~}      & \multicolumn{2}{c}{~}         & 1  & 5.88                   \\ \midrule
	LTL-as-LTLf/alaska/lift/lift:b:l              & 17                        & 14  & 82.35                & 17   & 100.00                & 3   & 17.65                & \multicolumn{2}{c}{~}          & \multicolumn{2}{c}{~}       & \multicolumn{2}{c}{~}          & \multicolumn{2}{c}{~}      & \multicolumn{2}{c}{~}         & \multicolumn{2}{c}{~}       \\ \midrule
	LTL-as-LTLf/alaska/lift/lift:f                & 17                        & 16  & 94.12                & 17   & 100.00                & 1   & 5.88                 & \multicolumn{2}{c}{~}          & \multicolumn{2}{c}{~}       & \multicolumn{2}{c}{~}          & \multicolumn{2}{c}{~}      & \multicolumn{2}{c}{~}         & \multicolumn{2}{c}{~}       \\ \midrule
	LTL-as-LTLf/alaska/lift/lift:f:l              & 17                        & 15  & 88.24                & 16   & 94.12                 & 1   & 5.88                 & \multicolumn{2}{c}{~}          & \multicolumn{2}{c}{~}       & \multicolumn{2}{c}{~}          & \multicolumn{2}{c}{~}      & \multicolumn{2}{c}{~}         & 1  & 5.88                   \\ \midrule
	LTL-as-LTLf/alaska/lift/lift:l                & 17                        & 16  & 94.12                & 17   & 100.00                & 1   & 5.88                 & \multicolumn{2}{c}{~}          & \multicolumn{2}{c}{~}       & \multicolumn{2}{c}{~}          & \multicolumn{2}{c}{~}      & \multicolumn{2}{c}{~}         & \multicolumn{2}{c}{~}       \\ \midrule
	LTL-as-LTLf/anzu/amba/amba:c                  & 17                        & 15  & 88.24                & 15   & 88.24                 & \multicolumn{2}{c}{~}      & \multicolumn{2}{c}{~}          & \multicolumn{2}{c}{~}       & \multicolumn{2}{c}{~}          & \multicolumn{2}{c}{~}      & \multicolumn{2}{c}{~}         & 2  & 11.76                  \\ \midrule
	LTL-as-LTLf/anzu/amba/amba:cl                 & 17                        & 16  & 94.12                & 16   & 94.12                 & \multicolumn{2}{c}{~}      & \multicolumn{2}{c}{~}          & \multicolumn{2}{c}{~}       & \multicolumn{2}{c}{~}          & \multicolumn{2}{c}{~}      & \multicolumn{2}{c}{~}         & 1  & 5.88                   \\ \midrule
	LTL-as-LTLf/anzu/genbuf/genbuf                & 20                        & 20  & 100.00               & 20   & 100.00                & \multicolumn{2}{c}{~}      & \multicolumn{2}{c}{~}          & \multicolumn{2}{c}{~}       & \multicolumn{2}{c}{~}          & \multicolumn{2}{c}{~}      & \multicolumn{2}{c}{~}         & \multicolumn{2}{c}{~}       \\ \midrule
	LTL-as-LTLf/anzu/genbuf/genbuf:c              & 20                        & 20  & 100.00               & 20   & 100.00                & \multicolumn{2}{c}{~}      & \multicolumn{2}{c}{~}          & \multicolumn{2}{c}{~}       & \multicolumn{2}{c}{~}          & \multicolumn{2}{c}{~}      & \multicolumn{2}{c}{~}         & \multicolumn{2}{c}{~}       \\ \midrule
	LTL-as-LTLf/anzu/genbuf/genbuf:cl             & 20                        & 18  & 90.00                & 18   & 90.00                 & \multicolumn{2}{c}{~}      & \multicolumn{2}{c}{~}          & \multicolumn{2}{c}{~}       & \multicolumn{2}{c}{~}          & \multicolumn{2}{c}{~}      & \multicolumn{2}{c}{~}         & 2  & 10.00                  \\ \midrule
	LTL-as-LTLf/forobots                          & 38                        & 38  & 100.00               & 38   & 100.00                & \multicolumn{2}{c}{~}      & \multicolumn{2}{c}{~}          & \multicolumn{2}{c}{~}       & \multicolumn{2}{c}{~}          & \multicolumn{2}{c}{~}      & \multicolumn{2}{c}{~}         & \multicolumn{2}{c}{~}       \\ \midrule
	LTL-as-LTLf/rozier/counter/counter            & 19                        & \multicolumn{2}{c}{~}      & \multicolumn{2}{c}{~}        & \multicolumn{2}{c}{~}      & \multicolumn{2}{c}{~}          & \multicolumn{2}{c}{~}       & \multicolumn{2}{c}{~}          & 14  & 73.68                & 14 & 73.68                    & 5  & 26.32                  \\ \midrule
	LTL-as-LTLf/rozier/counter/counterCarry       & 19                        & \multicolumn{2}{c}{~}      & \multicolumn{2}{c}{~}        & 6   & 31.58                & \multicolumn{2}{c}{~}          & \multicolumn{2}{c}{~}       & \multicolumn{2}{c}{~}          & 12  & 63.16                & 18 & 94.74                    & 1  & 5.26                                 \\ \midrule
	LTL-as-LTLf/rozier/counter/counterCarryLinear & 19                        & \multicolumn{2}{c}{~}      & \multicolumn{2}{c}{~}        & 6   & 31.58                & \multicolumn{2}{c}{~}          & \multicolumn{2}{c}{~}       & \multicolumn{2}{c}{~}          & 13  & 68.42                & 19 & 100.00                   & \multicolumn{2}{c}{~}       \\ \midrule
	LTL-as-LTLf/rozier/counter/counterLinear      & 18                        & \multicolumn{2}{c}{~}      & 1    & 5.56                  & \multicolumn{2}{c}{~}      & \multicolumn{2}{c}{~}          & \multicolumn{2}{c}{~}       & \multicolumn{2}{c}{~}          & 18  & 100.00               & 17 & 94.44                    & \multicolumn{2}{c}{~}       \\ \midrule
	LTL-as-LTLf/rozier/formulas/n                 & 30                        & 12  & 40.00                & 24   & 80.00                 & \multicolumn{2}{c}{~}      & \multicolumn{2}{c}{~}          & 8  & 26.67                  & 4 & 13.33                      & 10  & 33.33                & 2  & 6.67                     & \multicolumn{2}{c}{~}       \\ \midrule
	LTL-as-LTLf/schuppan/O1formula                & 27                        & 20  & 74.07                & 20   & 74.07                 & \multicolumn{2}{c}{~}      & \multicolumn{2}{c}{~}          & 4  & 14.81                  & 4 & 14.81                      & 3   & 11.11                & 3  & 11.11                    & \multicolumn{2}{c}{~}       \\ \midrule
	LTL-as-LTLf/schuppan/O2formula                & 27                        & 5   & 18.52                & 5    & 18.52                 & \multicolumn{2}{c}{~}      & \multicolumn{2}{c}{~}          & \multicolumn{2}{c}{~}       & \multicolumn{2}{c}{~}          & 8   & 29.63                & 8  & 29.63                    & 14 & 51.85                   \\ \midrule
	LTL-as-LTLf/schuppan/phltl                    & 13                        & 12  & 92.31                & 12   & 92.31                 & \multicolumn{2}{c}{~}      & \multicolumn{2}{c}{~}          & \multicolumn{2}{c}{~}       & \multicolumn{2}{c}{~}          & \multicolumn{2}{c}{~}      & \multicolumn{2}{c}{~}         & 1  & 7.69                    \\ \midrule
	LTLf-specific/benchmarks:ltlf/\ldots/C100     & 500                       & 255 & 51.00                & 500  & 100.00                & 245 & 49.00                & \multicolumn{2}{c}{~}          & \multicolumn{2}{c}{~}       & \multicolumn{2}{c}{~}          & \multicolumn{2}{c}{~}      & \multicolumn{2}{c}{~}         & \multicolumn{2}{c}{~}       \\ \midrule
	LTLf-specific/benchmarks:ltlf/\ldots/V20      & 425                       & 127 & 29.88                & 425  & 100.00                & 272 & 64.00                & \multicolumn{2}{c}{~}          & \multicolumn{2}{c}{~}       & \multicolumn{2}{c}{~}          & 26  & 6.12                 & \multicolumn{2}{c}{~}         & \multicolumn{2}{c}{~}          \\ \arrayrulecolor{black} \bottomrule
\end{tabular}
  \end{adjustbox}
  \endgroup
\end{table}
%


\section{Related work}
\label{sec:related}

\cdc{To the best of our knowledge, this is the first research endeavour aimed at extrating unsatisfiable cores for \LTLf.} 
In the following, we review the most relevant
literature concerning \LTL/\LTLf satisfiability, and \LTL SAT-based \UC
extraction.

The \LTL satisfiability problem has been addressed through tableau-based methods (e.g., \cite{Janssen99}), temporal resolution (e.g., \cite{DBLP:journals/tocl/FisherDP01}), and reduction to model checking (e.g., \cite{DBLP:conf/cav/CimattiRST07,DBLP:journals/sttt/RozierV10,DBLP:conf/fm/RozierV11}).
In~\cite{DBLP:journals/sttt/RozierV10}, a reduction of the \LTL satisfiability problem to a model checking problem, and a comparison of different model checkers (explicit/symbolic) has been carried out, resulting in better performance and quality for symbolic approaches.
A thorough comparison of the main tools dealing with the \LTL satisfiability problem is reported in~\cite{DBLP:conf/atva/SchuppanD11}. The paper considers also
tableau and temporal resolution based solvers, revealing a complementary behaviour between some of the considered solvers.

The problem of checking the satisfiability of \LTLf properties has been the subject of several works~\cite{DBLP:conf/ecai/Li0PVH14,
DBLP:journals/jair/FiondaG18,DBLP:journals/ai/LiPZVR20}.
\citet{DBLP:conf/ecai/Li0PVH14}, leverage the finite semantics of traces for introducing a propositional SAT based algorithm for the \LTLf satisfiability problem together with some heuristics to guide the search. The approach has been implemented in the \textsc{aalta-finite} tool, which has been shown to outperform other existing approaches based on the reduction to the \LTL satisfiability problem.
This work has been then extended in~\cite{DBLP:journals/ai/LiPZVR20} to leverage a transition system (TS) for the input \LTLf formula, and reducing satisfiability checking to a SAT based path-search problem over this TS. This approach, also implemented in \textsc{aalta-finite}, has been shown to provide the best results in checking unsatisfiable formulae and comparable results for satisfiable ones.
%
\citet{DBLP:journals/jair/FiondaG18} investigate the complexity of some fragments of \LTLf, and present a SAT based algorithm that outperforms the \textsc{aalta-finite} version in~\cite{DBLP:conf/ecai/Li0PVH14}.
%
Algorithm 3 presented here exploits the work in \cite{DBLP:journals/ai/LiPZVR20} as state-of-the-art tool for checking the satisfiability of \LTLf properties.

The \UC extraction for \LTL has also been the subject of several studies~\cite{gore_huang_sergeant_thomson_mus_pltl,DBLP:conf/caise/AwadGTW11,DBLP:journals/acta/Schuppan16,DBLP:conf/nfm/NarizzanoPTV18}.
\citet{gore_huang_sergeant_thomson_mus_pltl} presents a BDD based approach that leverages a method to determine minimal \UCs for SAT~\cite{DBLP:conf/aspdac/Huang05} to find minimal \UCs in \LTL. 
In~\cite{DBLP:conf/caise/AwadGTW11}, \UCs are extracted by leveraging a tableau-based solver to obtain an initial subset of unsatisfiable \LTL formulae and then applying a deletion-based minimization to the subset. The approach, implemented in \textsc{procmine} is part of a tool for the synthesis of business process templates.
In~\cite{DBLP:journals/acta/Schuppan16} fine-grained \UCs are extracted constructing and optimizing resolution graphs for temporal resolution.
Finally, \citet{DBLP:conf/nfm/NarizzanoPTV18} presents a SAT based encoding suitable for the unsat core extraction of \LTL-based property specification patterns~\cite{DBLP:conf/icse/DwyerAC99} extended with inequality statements on Boolean and numeric variables.
Algorithm 4 presented here starts from the work in \cite{DBLP:journals/acta/Schuppan16} to compute \UCs using temporal resolution.

In the context of process mining, the works in \cite{DiCiccio.etal/IS2017:ResolvingInconsistenciesRedundanciesDeclare} and \cite{DBLP:conf/bpm/CoreaD19} identify inconsistencies for  specific \LTLf-based constraints contained in the Declare \cite{2009-Aalst} modeling language.
They rely on automata language and language inclusion techniques to identify the inconsistencies, and are specific to the precise structure of Declare. Thus they cannot be generalized to address generic \LTLf-based specifications.

Finally, works on propositional \UC extraction (e.g., \cite{DBLP:conf/aspdac/Huang05,DBLP:journals/corr/Marques-SilvaJ14,DBLP:conf/date/GoldbergN03}) could be used to improve the quality of the computed cores but we leave this investigation for future developments.

\section{Conclusions and future work}
\label{sec:conclusions}

In this paper, we have addressed the problem of \LTLf unsatisfiable core extraction, presenting
four algorithms based on different state-of-the-art techniques for \LTL and \LTLf satisfiability checking. We have implemented each of them based on existing tools, and we have carried out an experimental
evaluation on a set of \mr{reference} benchmarks for unsatisfiable temporal formulas. The results
have shown feasibility and complementarities of the proposed algorithms.

For future work, we \cdc{envisage the following research endeavours. We}
\begin{enumerate*}[(i)]
\item will address the problem of extracting minimal UCs,
\item plan to extend the approach to other \LTL/\LTLf algorithms based
  on $k$-liveness~\cite{kliveness}, liveness to
  safety~\cite{liveness2safety}, or tableau
  constructions~\cite{DBLP:journals/iandc/GeattiGMR21},
\item intend to extend the problems set with benchmarks from emerging
  domains (e.g., AI Planning, or BPM),
\item want to correlate structural information (e.g., $\PVarSet$
  cardinality, temporal depth, number of operators) with solving
  algorithms,
\item aim to investigate the extension to the infinite state case
  exploiting SMT
  techniques~\cite{DBLP:series/faia/BarrettSST09,DBLP:conf/cav/DanielCGTM16}.
\end{enumerate*}

\clearpage

\bibliography{biblio}

\end{document}